\documentclass[11pt, fleqn,leqno]{article}
\pdfoutput=1

\usepackage[margin=2.5cm]{geometry}
\usepackage{lipsum}

\usepackage[utf8]{inputenc}
\usepackage[T1]{fontenc}
\usepackage[colorlinks]{hyperref}
\usepackage{microtype}
\usepackage{xcolor}

\usepackage{enumitem}
\usepackage{comment}

\hypersetup{
  linkcolor=[rgb]{0.3,0.3,0.6},
  citecolor=[rgb]{0.2, 0.6, 0.2},
  urlcolor=[rgb]{0.6, 0.2, 0.2}
}

\usepackage{mathtools, amsthm, amsfonts, amssymb}

\usepackage{mathrsfs}
\usepackage{pbox}

\usepackage{booktabs}
\usepackage{braket}
\usepackage{commath}
\usepackage{caption}

\usepackage[capitalise]{cleveref}
\usepackage{accents}

\usepackage{tikz}
\usetikzlibrary{positioning}

\usepackage{tikz-cd}

\theoremstyle{plain}
\newtheorem{theorem}{Theorem}
\newtheorem*{theorem*}{Theorem}
\newtheorem{proposition}[theorem]{Proposition}
\newtheorem{lemma}[theorem]{Lemma}
\newtheorem{corollary}[theorem]{Corollary}

\theoremstyle{definition}
\newtheorem{definition}[theorem]{Definition}

\newtheorem{conjecture}[theorem]{Conjecture}
\newtheorem*{problem*}{Problem}

\newtheorem{remark}[theorem]{Remark}
\newtheorem{example}[theorem]{Example}

\DeclareMathOperator{\rank}{R}
\DeclareMathOperator{\subrank}{Q}

\DeclareMathAccent{\wtilde}{\mathord}{largesymbols}{"65}
\DeclareMathOperator{\asymprank}{\underaccent{\wtilde}{R}}
\DeclareMathOperator{\asympsubrank}{\underaccent{\wtilde}{Q}}

\newcommand{\Proj}[1]{\ket{#1}\!\bra{#1}}

\newcommand{\F}{\mathbb{F}}
\newcommand{\Z}{\mathbb{Z}}

\newcommand{\R}{\mathbb{R}}
\newcommand{\C}{\mathbb{C}}
\newcommand{\N}{\mathbb{N}}

\DeclareMathOperator{\rk}{rank}

\newcommand{\cD}{\mathcal{D}}
\newcommand{\cE}{\mathcal{E}}

\newcommand{\cG}{\mathcal{G}}
\newcommand{\cH}{\mathcal{H}}
\newcommand{\cK}{\mathcal{K}}
\newcommand{\cI}{\mathcal{I}}
\newcommand{\cL}{\mathcal{L}}
\newcommand{\cM}{\mathcal{M}}
\newcommand{\cN}{\mathcal{N}}

\newcommand{\cS}{\mathcal{S}}

\newcommand{\bX}{\mathbf{X}}

\newcommand{\linspan}{\mathrm{span}}

\newcommand{\EAto}{\overset{\smash{*}}{\to}}
\newcommand{\Qto}{\overset{\smash{q}}{\to}}

\DeclareMathOperator{\Hom}{Hom}

\newcommand{\changed}[1]{#1}

\newcommand{\asympleq}{\lesssim}

\DeclareMathOperator{\Tr}{Tr}

\let\Re\relax
\let\Im\relax
\DeclareMathOperator{\Re}{Re}
\DeclareMathOperator{\Im}{Im}

\title{Quantum asymptotic spectra of\\ graphs and non-commutative graphs,\\ and quantum Shannon capacities}
\author{Yinan Li \thanks{Centrum Wiskunde \& Informatica and QuSoft, Science Park 123, 1098XG Amsterdam, Netherlands ({\tt Yinan.Li@cwi.nl}). Supported by ERC Consolidator Grant 615307-QPROGRESS.}
\and 
Jeroen Zuiddam %
\thanks{Institute for Advanced Study, 1 Einstein Drive, Princeton, NJ 08540, USA ({\tt jzuiddam@ias.edu}). Supported by NWO (617.023.116), the QuSoft Research Center for Quantum Software and the National Science Foundation under Grant No.~DMS-1638352. }
}
\date{\today}

\begin{document}

\maketitle
\begin{abstract}
We study quantum versions of the Shannon capacity of graphs and non-commutative graphs.
We introduce the asymptotic spectrum of graphs with respect to quantum and entanglement-assisted homomorphisms, and we introduce the asymptotic spectrum of non-commutative graphs with respect to entanglement-assisted homomorphisms.
We apply Strassen's spectral theorem (J.~Reine Angew.\ Math., 1988) in order to obtain dual characterizations of the corresponding Shannon capacities and asymptotic preorders in terms of their asymptotic spectra.
This work extends the study of the asymptotic spectrum of graphs initiated by Zuiddam (Combinatorica, 2019) to the quantum domain.

\changed{We then exhibit spectral points in the new quantum asymptotic spectra and discuss their relations with the asymptotic spectrum of graphs. In particular, we prove that the (fractional) real and complex Haemers bounds upper bound the quantum Shannon capacity, which is defined as the regularization of the quantum independence number (Man\v{c}inska and Roberson, J.\ Combin.\ Theory Ser.\ B, 2016), and that the fractional real and complex Haemers bounds are elements in the quantum asymptotic spectrum of graphs. This is in contrast to the Haemers bounds defined over certain finite fields, which can be strictly smaller than the quantum Shannon capacity. Moreover, since the Haemers bound can be strictly smaller than the Lov{\'a}sz theta function (Haemers, IEEE Trans.\ Inf.\ Theory, 1979), we find that the quantum Shannon capacity and the Lov{\'a}sz theta function do not coincide. As a consequence, two well-known conjectures in quantum information theory, namely,
\begin{itemize}
\item The entanglement-assisted zero-error capacity of a classical channel is equal to the Lov\'asz theta function; and
\item Maximally entangled states and projective measurements are sufficient to achieve the entanglement-assisted zero-error capacity,
\end{itemize}
cannot both be true.}

\end{abstract}
\section{Introduction}
This paper studies quantum variations of the \emph{Shannon capacity} of graphs via the \emph{theory of asymptotic spectra}.
The Shannon capacity of a graph $G$, introduced by Shannon in~\cite{MR0089131}, is defined as 
\[
\Theta(G):=\sup_{n\geq 1}\sqrt[n]{\alpha(G^{\boxtimes n})}=\lim_{n\to\infty}\sqrt[n]{\alpha(G^{\boxtimes n})},
\]
where $\alpha(G)$ denotes the \emph{independence number} of~$G$ and where~$G^{\boxtimes n}$ denotes the $n$-th \emph{strong graph product} power of~$G$. (All concepts used in the introduction will be explicitly defined in Section~\ref{sec: preliminaries}.)
The definition of Shannon capacity is motivated by the study of \emph{classical communication channels}. One associates to a classical channel the \emph{confusability graph} with vertices given by the input symbols of the channel, and edges given by the pairs of input symbols that are mapped to the same output by the channel with a nonzero probability. %
The 
Shannon capacity of the confusability graph then measures the amount of information that can be transmitted over the channel without error, asymptotically. 

Deciding whether $\alpha(G) \geq k$ is NP-complete~\cite{MR0378476}, and Shannon capacity is not even known to be a computable function.
A natural approach to study the Shannon capacity is to construct graph parameters that are upper bounds on Shannon capacity. Shannon himself introduced an upper bound in~\cite{MR0089131}, which is known as the \emph{fractional packing number}, \emph{Rosenfeld number}~\cite{10.2307/2035288}, or \emph{fractional clique cover number}\footnote{The fractional packing number is defined in terms of the classical channel itself and equals the fractional clique cover number~\cite[Section~64.8]{schrijver2003combinatorial} of the corresponding confusability graph. See, e.g.,~the discussion in~\cite[Appendix A]{acin2017new}.}. 
In the seminal work of Lov\'asz~\cite{lovasz1979shannon}, the \emph{Lov\'asz theta function} $\vartheta$ was introduced to better upper bound the Shannon capacity. Remarkably, the theta function can be written as a semidefinite program that is efficiently computable. %
Using the theta function, Lov\'asz proved that $\alpha(C_5^{\boxtimes 2})^{1/2} = \Theta(C_5)= \vartheta(C_5) = \sqrt{5}$, where $C_n$ is the $n$-cycle graph. Lov\'asz further conjectured that $\Theta(G)=\vartheta(G)$ for every graph $G$. 
This conjecture was shown to be false by Haemers, who introduced the \emph{Haemers bound} $\cH^\F$ (over a field $\F$) as an upper bound on the Shannon capacity. The Haemers bound can be strictly smaller than the Lov\'asz theta function (e.g.,~on the complement of the Schl\"afli graph~\cite{haemers1979some}).

Even for the odd cycle graphs $C_{2k+1}$ with $k\geq 3$, it is still open whether $\Theta(C_{2k+1}) = \vartheta(C_{2k+1})$. For example, the currently best lower bound on $\Theta(C_7)$ is $\sqrt[5]{367}\approx 3.25787$~\cite{polak2018new}, whereas $\vartheta(C_7) \approx 3.31766$.

Recently, a \emph{dual characterization} of the Shannon capacity was found by Zuiddam in~\cite{zuiddam2018asymptotic} via the theory of asymptotic spectra. 
This theory was developed by Strassen in \cite{strassen1988asymptotic}. \changed{(See also the exposition in \cite[Chapter 1]{phd}.)}
In the general theory we are given a \emph{commutative semiring}~$S$ with addition $+$, multiplication $\cdot$, and a preorder~$\leq$ on $S$ that satisfies the properties to be a ``Strassen preorder'' (see Definition~\ref{strassen preorder}). %
For $a \in S$, the \emph{rank} $\rank(a)$ is defined as the minimum number $n$ such that~$a\leq n$, and the \emph{subrank}~$\subrank(a)$ is defined as the maximum number $n$ such that~$n\leq a$, where $n\in S$ stands for the sum of~$n$ times the element~$1 \in S$.
The \emph{asymptotic rank} of $a$ is defined as the regularization $\lim_{n\to\infty}\sqrt[n]{\rank(a^n)}$ and the \emph{asymptotic subrank} as the regularization $\lim_{n\to\infty}\sqrt[n]{\subrank(a^n)}$. %
\changed{The \emph{asymptotic spectrum} $\bX(S,\leq)$ of $S$ with respect to $\leq$ is the set of all $\leq$-\emph{monotone semiring homomorphisms} $S \to \R_{\geq 0}$. Strassen proved that the asymptotic spectrum $\bX(S,\leq)$ characterizes the ``asymptotic preorder'' $\asympleq$ induced by a Strassen preorder $\leq$. We have $a \asympleq b$ if there exists a sequence $(x_n)_{n\in \N} \subseteq \N$ such that $\inf_n x_n^{1/n} = 1$ and such that for all~$n \in \N$, the relation $a^n \leq x_n \cdot b^n$ holds. This leads to the following dual characterizations of the asymptotic rank and subrank (\cite[Theorem 3.8]{strassen1988asymptotic} and \cite[Corollary~2.14]{phd}).
Let $S$ be a commutative semiring and $\leq$ be a Strassen preorder. 
\begin{itemize}
\item For any $a\in S$ satisfying that there exists $\phi\in\bX(S,\leq)$ such that $\phi(a)\geq 1$, we have
$$\lim_{n\to\infty}\sqrt[n]{\rank(a^n)}=\max\{\phi(a):~\phi\in\bX(S,\leq)\}.$$
\item For any $b\in S$ satisfying that there exists $k\in\N$ such that $b^k\geq 2$, we have
$$\lim_{n\to\infty}\sqrt[n]{\subrank(b^n)}=\min\{\phi(b):~\phi\in\bX(S,\leq)\}.$$
\end{itemize}}

The theory of asymptotic spectra was originally motivated by the study of \emph{tensor rank} and \emph{asymptotic tensor rank}~\cite{Strassen:1986:AST:1382439.1382931,strassen1987relative,strassen1988asymptotic,strassen1991degeneration}, which are the keys to understand the arithmetic complexity of matrix multiplication~(see, e.g., \cite{burgisser1997algebraic}). 
Here we let $S$ be any family of isomorphism classes of tensors 
\changed{(under local general linear group actions on local spaces)}
that is closed under direct sum and tensor product, and which contains the ``diagonal tensors''. %
We let $\leq$ be the restriction preorder, which in quantum information theory language is the preorder corresponding to convertibility by stochastic local operations and classical communication~(SLOCC).
The restriction preorder is a Strassen preorder, the rank as defined above equals tensor rank, and the asymptotic rank as defined above equals asymptotic tensor rank.
\changed{Recently, Christandl, Vrana and Zuiddam in~\cite{christandl2017universalproc} constructed an infinite family of elements in the asymptotic spectrum of tensors over the complex numbers.} A study of tensors with respect to local operations and classical communication was carried out in \cite{jensen2018asymptotic}.

Let us return to the study of graphs as in~\cite{zuiddam2018asymptotic}. Here $S$ is any family of isomorphism classes of graphs 
\changed{(under permutation group actions on vertices)}  
that is closed under the disjoint union and the strong graph product, and which contains the $n$-vertex empty graph~$\overline{K_n}$ for all $n\in\N$. Let $\leq$ be the \emph{cohomomorphism preorder}, which is defined by letting $G\leq H$ if there is a \emph{graph homomorphism} from the \emph{complement} of $G$ to the complement of $H$. 
\changed{Then the subrank of a graph equals the independence number and the rank of a graph equals the \emph{clique cover number} (which is the same as the \emph{chromatic number of the complement graph}). Furthermore, the asymptotic subrank coincides with the Shannon capacity and the asymptotic rank coincides with the fractional clique cover number.\footnote{The regularization of the clique cover number equals the fractional clique cover number~\cite{LOVASZ1975383}. See also~\cite[Theorem 67.17]{schrijver2003combinatorial}.} Zuiddam proved that the cohomomorphism preorder is a Strassen preorder~\cite{zuiddam2018asymptotic}. Thus the Shannon capacity (resp.~fractional clique cover number) equals the pointwise minimum (resp.~maximum) over the asymptotic spectrum of graphs.}
Known elements in the asymptotic spectrum of graphs are the Lov\'asz theta function $\vartheta$~\cite{lovasz1979shannon}, the \emph{fractional Haemers bound} $\cH_f^\F$ over any field $\F$~\cite{blasiak2013graph,bukh2018fractional}, the complement of the \emph{projective rank} $\overline{\xi}_f$~\cite{manvcinska2016quantum} and the fractional clique cover number $\overline{\chi}_f$~(see~\cite[Sec. $64.8$]{schrijver2003combinatorial})
\changed{Interestingly, for every field $\F$ of nonzero characteristic and $\epsilon>0$, there exist an explicit graph $G(\F,\epsilon)$ such that for every field $\F'$ of different characteristic, it holds that $\cH_f^{\F}(G(\F,\epsilon))<\epsilon\cH_f^{\F'}(G(\F,\epsilon))$~\cite[Theorem 19]{bukh2018fractional}. We thus know that there is an infinite family of elements in the asymptotic spectrum of graphs.}
\subsubsection*{Quantum Shannon capacity of graphs}

We now turn to the quantum setting. We consider two quantum variants of graph homomorphism. The first variant is characterized by the existence of perfect quantum strategies for the \emph{graph homomorphism game}~\cite{manvcinska2016quantum}, which is defined as follows. Two players Alice and Bob are given two graphs $G$ and $H$. During the game, the referee sends to Alice some vertex $g_A \in V(G)$ and to Bob some vertex $g_B \in V(G)$. Alice responds to the referee with a vertex $h_A \in V(H)$ and Bob \changed{responds} to the referee with a vertex~$h_B \in V(H)$. Alice and Bob win this instance of the $(G,H)$-homomorphism game, when their answer \changed{satisfies} the following two conditions:
\begin{gather}
\textnormal{if}~g_A=g_B,~{\rm then}~h_A=h_B\\
\changed{\textnormal{if}}~\{g_A,g_B\}\in E(G),~{\rm then}~\{h_A,h_B\}\in E(H).
\end{gather}
Alice and Bob are not allowed to communicate with each other after having received their input from the referee, but they may together decide on a strategy beforehand. 
It is not hard to see that Alice and Bob can win the $(G,H)$-homomorphism game with a classical strategy (i.e.~not sharing entangled states) if and only if there is a graph homomorphism from $G$ to~$H$. 
We say that there is a \emph{quantum homomorphism} from $G$ to $H$, and write $G\Qto H$, if Alice and Bob can win the $(G,H)$-homomorphism game using some shared entangled state.\footnote{It can be proved that if there is a perfect quantum strategy for the $(G,H)$-homomorphism game, it can be achieved using the maximally entangled state and projective measurements~\cite{manvcinska2016quantum}.} 
\changed{More precisely, $G\Qto H$ if there exist a $d\in \N$ and $d \times d$ projectors~$E_g^h\in M(d,\C)$ for every $g\in V(G)$ and $h\in V(H)$, such that the following two conditions hold:
\begin{gather}
\textnormal{for every $g\in V(G)$ we have $\textstyle\sum_{h\in V(H)} E_{g}^h=I_d$}\\
\textnormal{if $ \{g,g'\}\in E(G)$ and $\{h,h'\}\not\in E(H)$, then $E_g^h E_{g'}^{h'}=0$.}
\end{gather}}
The \emph{quantum cohomomorphism preorder} $\leq_q$ is defined by letting~$G\leq_q H$ if $\overline{G}\Qto\overline{H}$. The~\emph{quantum independence number} $\alpha_q(G)$ of $G$ is defined as the maximum number $n$ such that~$\overline{K_n}\leq_q G$ and the \emph{quantum Shannon capacity} $\Theta_q(G)$ of $G$ is defined as its regularization.

\subsubsection*{Entanglement-assisted Shannon capacity of graphs}
The second quantum variant of graph homomorphism comes from the study of \emph{entanglement-assisted} zero-error capacity of classical channels, which is a quantum generalization of Shannon's zero-error communication setting. In the zero-error communication model, Alice wants to transmit messages to Bob without error through some classical noisy channel. As we have mentioned, Shannon in~\cite{MR0089131} showed that the maximum number of zero-error messages Alice can send to Bob equals the independence number of the confusability graph. 
In the entanglement-assisted setting, the maximum number of messages that can be sent with zero error turns out to be also determined by the confusability graph and is called the \emph{entanglement-assisted independence number} (of the confusability graph)~\changed{\cite{cubitt2010improving}}. %
We define the \emph{entanglement-assisted Shannon capacity} as its regularization.
\changed{One can define a natural %
\emph{entanglement-assisted homomorphism} between graphs~\cite{cubitt2014bounds}, denoted by $G \EAto H$, 
by saying that $G \EAto H$ if there exists a $d\in \N$ and $d\times d$ positive semidefinite matrices $\rho$ and $(\rho_g^h\in M(d,\C): g\in V(G),\, h\in V(H))$, %
such that the following two conditions hold
\begin{gather}
\textnormal{for every $g\in V(G)$ we have $\textstyle\sum_{h\in V(H)} \rho_{g}^h=\rho$}\\
\textnormal{if $\{g,g'\}\in E(G)$ and $\{h,h'\}\not\in E(H)$, then $\rho_g^h \rho_{g'}^{h'}=0$}.
\end{gather}}
Let the \emph{entanglement-assisted cohomomorphism preorder} $\leq_*$ be defined by letting~$G\leq_* H$ if~$\overline{G}\EAto\overline{H}$. 
The entanglement-assisted independence number $\alpha_*(G)$ of $G$ can be equivalently defined as the maximum number $n$ such that~$\overline{K_n}\leq_* G$.
\changed{It is not hard to see that $G\leq H$ implies $G\leq_q H$ and $G\leq_q H$ implies $G\leq_* H$, which immediately implies that $\alpha(G)\leq\alpha_q(G)\leq\alpha_*(G)$ and $\Theta(G)\leq\Theta_q(G)\leq\Theta_*(G)$. It has been shown that there exist $G$ and $H$ such that $\alpha(G)<\alpha_q(G)$~\cite{cubitt2010improving,6466384} and $\Theta(H)<\Theta_q(H)$~\cite{Leung2012,briet2015entanglement}, which indicate that $\leq$ and $\leq_q$ are not the same preorder. It has been conjectured that the two quantum preorders $\leq_q$ and $\leq_*$ coincide, in which case one may interpret $G\Qto H$ in the communication setting as restricting to use {a} maximally entangled state and projective measurements~\cite{manvcinska2016quantum}.}

\subsubsection*{(Entanglement-assisted) Shannon capacity of non-commutative graphs}
Finally, we consider the setting of sending \emph{classical} zero-error messages through \emph{quantum channels}, \changed{where the input symbols are modeled by \emph{quantum states} (positive semidefinite matrices with trace~$1$) and the transition rules are modeled by \emph{completely positive and trace-preserving} (CPTP) maps.}
It turns out, analogous to the classical channel scenario, that the 
zero-error classical capacity of a quantum channel is characterized by the so called \emph{non-commutative graph} associated with the channel~\cite{duan2013}.
A non-commutative graph, or \emph{nc-graph} for short, is a subspace $S$ of the vector space of $n\times n$ complex matrices, satisfying $S^\dagger = S$ and $I \in S$.%
\footnote{In operator theory, nc-graphs are exactly operator systems. Duan in~\cite{duan2009super} and Cubitt, Chen and Harrow in~\cite{cubitt2011} have shown that every operator system is indeed associated to a quantum channel.} %
\changed{The independence number of an nc-graph $S$ is the largest integer $k$ such that there exist \emph{unit vectors} $\ket{\psi_1},\dots,\ket{\psi_k}$ such that $\ket{\psi_i}\!\bra{\psi_j}\perp S$ for any $i\neq j\in\{1,\dots,k\}$, where the orthogonality is with respect to the Gram-Schmidt inner product. The Shannon capacity $\Theta(S)$ of an nc-graph $S$ is given as the regularization: $\Theta(S)=\lim_{n\to\infty}\sqrt[n]{\alpha(S^{\otimes n})}$ ($\otimes$ stands for the Kronecker product). There is also a natural way to define the entanglement-assisted independence number and the entanglement-assisted Shannon capacity of nc-graphs. See for details Section~\ref{sec: ea nc graph}.}
There are natural preorders $\leq$ and $\leq_*$ on nc-graphs~\cite{stahlke2016} such that the independence number $\alpha(S)$ and the entanglement-assisted independence number $\alpha_*(S)$, defined in~\cite{duan2013}, equal the maximum number $n$ such that~$\overline{\cK_n} \leq S$ and the maximum number $n$ such that~$\overline{\cK_n} \leq_* S$~\cite{stahlke2016}, respectively. Here $\overline{\cK_n}$ is the nc-graph associated to the $n$-message perfect classical channel (whose confusability graph is $\overline{K_n}$).

\subsubsection*{Overview of our results}
In this paper, we extend the study of the asymptotic spectrum of graphs to the quantum domain. We introduce three new asymptotic spectra:
\begin{itemize}
\item the asymptotic spectrum of graphs $\bX(\cG,\leq_q)$ with respect to the quantum cohomomorphism preorder
\item the asymptotic spectrum of graphs $\bX(\cG,\leq_*)$ with respect to entanglement-assisted cohomomorphism preorder
\item the asymptotic spectrum of non-commutative graphs $\bX(\cS,\leq_*)$ with respect to the entanglement-assisted cohomomorphism preorder.
\end{itemize}
We prove that the preorders in these scenarios are Strassen preorders. This allows us to apply Strassen's spectral theorem to obtain characterizations of these asymptotic preorders in terms of their asymptotic spectra and dual characterizations of the corresponding Shannon capacities, respectively. 
\changed{We then exhibit elements in these quantum asymptotic spectra and study their relations with the asymptotic spectrum of graphs of~\cite{zuiddam2018asymptotic}. More precisely:
\begin{itemize}
\item We prove that the Lov\'asz theta function $\vartheta$ belongs to $\bX(\cG,\leq_*)$ (Theorem~\ref{thm: elements in leq_*}), which is contained in $\bX(\cG,\leq_q)$.
\item We prove that in addition, the complement of the projective rank and the fractional Haemers bounds $\cH_f^\C$ and $\cH_f^\R$ belong to $\bX(\cG,\leq_q)$, while the fractional clique cover number $\overline{\chi}_f$ and the fractional Haemers bound $\cH_f^{\F_q}$ over certain finite field $\F_q$ do not belong to $\bX(\cG,\leq_q)$ (Theorem~\ref{thm: element in leq q}).
\item We prove that a quantum version of the Lov\'asz theta function, introduced in~\cite{duan2013}, belongs to $\bX(\cS,\leq_*)$ (Theorem~\ref{quantum lovasz theta}). Moreover, there is a surjective map from $\bX(\cS,\leq_*)$ to $\bX(\cG,\leq_*)$ (Theorem~\ref{surje}).
\end{itemize}

Searching for elements in quantum asymptotic spectra also leads to new separation results between the quantum variants of Shannon capacity and known upper bounds. Remarkably, we derive that the Haemers bounds over $\R$ and $\C$ upper bound the quantum Shannon capacity, while the Haemers bounds over certain finite fields could be strictly smaller than that. Moreover, it is known that the complement of the Schl\"afli graph $G$ satisfies $\cH^\R(G)\leq 7<9=\vartheta(G)$~\cite{haemers1979some}. Thus, we separate the quantum Shannon capacity from the Lov\'asz theta function. This result connects two conjectures in quantum zero-error information theory. Namely, it has been conjectured that the quantum Shannon capacity equals the entanglement-assisted Shannon capacity~\cite{manvcinska2016quantum}, and that the entanglement-assisted Shannon capacity equals the Lov\'asz theta function~\cite{Beigi2010,5961832,Leung2012,duan2013,6466384,cubitt2014bounds,8260566}. Our results show that these two conjectures cannot both be true.\footnote{In the nc-graph setting, it has been shown in~\cite{8260566} that the entanglement-assisted Shannon capacity of nc-graphs can be \emph{strictly smaller} than a quantum generalization of the Lov\'asz theta function introduced in~\cite{duan2013}. See a detailed discussion in Section~\ref{sec: dual characterization nc graphs}.}  
}

\subsubsection*{Organization of this paper}
\changed{In Section~\ref{sec: preliminaries} we cover
\begin{itemize}
\item the basic definitions of graph theory, the Lov{\'a}sz theta function, the fractional Haemers bounds, the projective rank and the fractional clique cover number
\item the theory of asymptotic spectra of Strassen and its application to graphs
\item the definition and properties of the quantum variants of graph homomorphisms
\item the definition and properties of non-commutative graphs and their homomorphisms.
\end{itemize}}
In Section~\ref{sec: dual characterization graphs} we study the quantum Shannon capacity and the entanglement-assisted Shannon capacity via the corresponding asymptotic spectra. 
In Section~\ref{sec: dual characterization nc graphs} we study the entanglement-assisted Shannon capacity of non-commutative graphs via the corresponding asymptotic spectrum.\section{Preliminaries}\label{sec: preliminaries}
\subsection{Graphs, independence number, and Shannon capacity}
In this paper we consider only finite simple graphs, so graph will mean finite simple graph. 
For a graph $G$, we use $V(G)$ to denote the vertex set of $G$ and $E(G)$ to denote the edge set of $G$. 
We write~$\{g,g'\}\in E(G)$ to denote an edge between vertex $g$ and $g'$. Since our graphs are simple, $\{g,g'\} \in E(G)$ implies that $g\neq g'$. 
The complement of $G$ is the graph~$\overline{G}$ with $V(\overline{G})=V(G)$ and $E(\overline{G})=\{\{g,g'\}:\{g,g'\}\not\in E(G)~{\rm and}~g\neq g'\}$. 
(We emphasize that when we write~$\{g, g'\} \not\in E(G)$ we include the case that $g = g'$.)
For $n\in\N$, the complete graph~$K_n$ is the graph with $V(K_n)=[n]:=\{1,2,\dots,n\}$ and $E(K_n) = \{\{i,j\}: i\neq j\in[n]\}$. 
Thus $K_0 = \overline{K_0}$ is the empty graph and $K_1 = \overline{K_1}$ is the graph consisting of a single vertex and no edges. A graph homomorphism from $G$ to $H$ is a map $f: V(G)\to V(H)$, such that for all~$g,g' \in V(G)$, $\{g,g'\}\in E(G)$ implies $\{f(g),f(g')\}\in E(H)$. We write $G\to H$ if there exists a graph homomorphism from $G$ to $H$.

A clique of $G$ is a subset $C$ of $V(G)$, such that for every $g\neq g'\in C$ \changed{we have} $\{g,g'\}\in E(G)$.
The size of the largest clique of $G$ is called the clique number of $G$ and is denoted by $\omega(G)$. Equivalently, 
\begin{equation}\label{eq: clique number}
\omega(G)=\max\{n\in\N:K_n\to G\}.
\end{equation}
An independent set of $G$ is a clique of $\overline{G}$. The size of the largest independent set of $G$ is called the independence number of $G$ and is denoted by $\alpha(G)$. Equivalently,  
\begin{equation}\label{eq: independence number}
\alpha(G)=\max\{n\in\N:K_n\to \overline{G}\}.
\end{equation}
Let $G$ and $H$ be graphs.
The disjoint union $G\sqcup H$ is the graph with $V(G\sqcup H)=V(G)\sqcup V(H)$ and $E(G\sqcup H)=E(G)\sqcup E(H)$. The strong graph product $G\boxtimes H$ is the graph with 
\begin{gather*}
V(G\boxtimes H)=V(G)\times V(H):=\{(g,h):g\in V(G),\, h\in V(H)\}\\
\begin{split}
E(G\boxtimes H)=\{\{(g,h),(g',h')\}:&~(g=g'~{\rm and}~\{h,h'\}\in E(H))\\
&{\rm or}~(\{g,g'\}\in E(G)~{\rm and}~\{h,h'\}\in E(H))\\
&{\rm or}~(\{g,g'\}\in E(G)~{\rm and}~h=h')\}.
\end{split}
\end{gather*}
We use $G^{\boxtimes N}$ to denote $\underbrace{G\boxtimes\cdots\boxtimes G}_N$. 
The Shannon capacity of $G$~\cite{MR0089131} is defined as 
\begin{equation}\label{eq: Shannon Capacity of graphs}
\Theta(G)\coloneqq\lim_{N\to\infty}\sqrt[N]{\alpha(G^{\boxtimes N})}.
\end{equation}
This limit exists and equals the supremum $\sup_N \sqrt[N]{\alpha(G^{\boxtimes N})}$ by Fekete's lemma.

\subsection{Upper bounds on the Shannon capacity} %
For any $d \in \N$ and any field $\F$, let $M(d, \F)$ be the space of $d\times d$ matrices with coefficients in~$\F$.
For $A \in M(d, \F)$ we let $A^T$ denote the transpose of $A$.
Let $I_d \in M(d, \F)$ be the $d\times d$ identity matrix. 
Let $A \in M(d, \C)$.
Then $A^\dagger$ denotes the complex conjugate of $A$. %
The element $A$ is called a projector if $A^\dagger = A$ and $AA = A$, that is, if~$A$ is Hermitian and idempotent.

Deciding whether $\alpha(G) \geq k$ is NP-hard~\cite{MR0378476} and it is not known whether the Shannon capacity~$\Theta(G)$ is a computable function. In the study of $\Theta(G)$, the following graph parameters have been introduced that upper bound $\Theta(G)$.

\subsubsection*{Lov\'asz theta function $\vartheta(G)$} An orthonormal representation of a graph $G$ is a collection of unit vectors $U=(u_g\in\R^d:g\in V(G))$ indexed by the vertices of $G$, such that non-adjacent vertices receive orthogonal vectors: $u^T_gu_{g'}=0$ for all $g\neq g', \{g,g'\}\not\in E(G)$. The celebrated Lov\'asz theta function~\cite{lovasz1979shannon}, is defined as
\begin{equation}\label{eq: lovasz theta}
\vartheta(G):=\min_{c,U}\max_{g\in V(G)}\frac{1}{(c^Tu_g)^2},
\end{equation}
where the minimization goes over unit vectors $c \in \R^d$ and orthonormal representations $U$ of $G$.
Lov\'asz proved that 
\[
\Theta(G) \leq \vartheta(G).
\]
Equation~\eqref{eq: lovasz theta} is a semidefinite program that is efficiently computable. There are several useful alternative characterizations of $\vartheta$ in the literature~\changed{\cite{lovasz1979shannon}}.
\subsubsection*{Fractional Haemers bound $\cH_f^\F(G)$} A $d$-representation of a graph $G$ over a field $\F$ is a matrix $M\in M(|V(G)|,\F)\otimes M(d,\F)$ of the form $M=\sum_{g,g'\in V(G)}e_ge_{g'}^{\dagger}\otimes M_{g,g'}$, such that $M_{g,g}=I_d$ for all $g\in V(G)$ and $M_{g,g'}=0$ if $g\neq g', \{g,g'\}\not\in E(G)$. Let $\cM_\F^d(G)$ be the set of all $d$-representation of $G$ over $\F$. The fractional Haemers bound~\cite{blasiak2013graph,bukh2018fractional}, as a fractional version of the Haemers bound~\cite{haemers1979some}, is defined as
\begin{equation}\label{eq: fractional Haemers bound}
\cH_f^\F(G):=\inf \bigl\{\rk(M)/ d : M\in\cM_\F^d(G),~d\in\N \bigr\}.
\end{equation}
The original Haemers bound~\cite{haemers1979some} of a graph $G$ can be formulated as:
\begin{equation}\label{eq: hamears bound}
\cH^\F(G)=\min \bigl\{\rk(M) : M\in\cM_\F^1(G)\bigr\}
\end{equation}
and we have
\[
\Theta(G) \leq \cH_f^\F(G)\leq\cH^\F(G).
\]
Whether the (fractional) Haemers bound is computable remains unknown. Interestingly, for any field $\F$ of nonzero characteristic and $\epsilon>0$, there exists an explicit graph $G = G(\F, \epsilon)$ so that if $\F'$ is any field with a different characteristic, $\cH_f^\F(G)\leq\epsilon\cH_f^{\F'}(G)$~\cite[Theorem~19]{bukh2018fractional}. %

\subsubsection*{Projective rank $\xi_f(G)$} A $d/r$-representation of a graph~$G$ is a collection of rank-$r$ projectors $(E_g\in M(d,\C): g\in V(G))$, such that %
$E_gE_{g'}=0$ if $\{g,g'\}\in E(G)$. The projective rank~\cite{manvcinska2016quantum} is defined as
\begin{equation}\label{eq: projective rank}
\xi_f(G):=\inf \bigl\{d/r:G~{\rm has~a}~d/r~{\rm representation}\bigr\}.
\end{equation}
The complement of the projective rank, $\overline{\xi}_f(G) \coloneqq \xi_f(\overline{G})$, is an upper bound on the Shannon capacity,
\[
\Theta(G) \leq \overline{\xi}_f(G).
\]

\subsubsection*{Fractional clique cover number $\overline{\chi}_f(G)$}
The fractional packing number can be written as a linear program (of large size), whose dual program is the fractional clique cover number (see, e.g.,~\cite{schrijver2003combinatorial} or~\cite[Eq.~(A.16)]{acin2017new}).  
Explicitly,
\begin{equation}
\begin{split}
\overline{\chi}_f(G):&=\min\sum_{C} s_C,~\text{s.t.}~s_C\geq 0~\text{for every clique }C,~\sum_{C\ni g}s_C\geq 1~\text{for every vertex }g\in V(G), \\
&=\max\sum_{g} t_g~\text{s.t.}~t_g\geq 0~\text{for every vertex }g\in V(G)~\sum_{g\in C}t_g\leq 1~\text{for every clique }C.
\end{split}
\end{equation}
where a clique $C$ of $G$ is an independent set of $\overline{G}$.
It is known that 
\begin{equation}\label{eq: fractional clique cover}
\Theta(G) \leq \overline{\chi}_f(G)=\lim_{n\to\infty}\sqrt[n]{\overline{\chi}(G^{\boxtimes n})}~(\text{e.g.~see}~\cite{schrijver2003combinatorial}).
\end{equation}

\subsection*{Relations between graph parameters}
We know the following inequalities among the graph parameters that we have just defined:
\begin{gather}
\Theta(G) \leq \vartheta(G) \leq \overline{\xi}_f(G) \leq \overline{\chi}_f(G)\label{ineq1}\\
\Theta(G) \leq \cH^\F_f(G) \leq \overline{\chi}_f(G)\label{ineq2}\\
\cH^\C_f(G) \leq \cH^\R_f(G) \leq \overline{\xi}_f(G).\label{ineq3}
\end{gather}
The inequalities in \eqref{ineq1} can be found in~\cite{lovasz1979shannon,manvcinska2016quantum}.
The inequalities in \eqref{ineq2} follow from the work in~\cite{bukh2018fractional}.
The first inequality in \eqref{ineq3} is actually an equality (cf. Prop.~\ref{prop: real and complex}), and the argument that the real fractional Haemers bound is at most the complement of the real projective rank $\overline{\xi}^{\smash{\R}}(G)$ is the following: We can obtain the definition of $\overline{\xi}^{\smash{\R}}_{\smash{f}}(G)$ from the definition of $\cH^{\smash{\R}}_{\smash{f}}(G)$ by requiring the $d$-representations of $G$ to be positive semidefinite, as implicitly shown in~\cite{Hogben17}.

\subsection{Asymptotic spectra and Strassen's spectral theorem}
We present some fundamental abstract concepts and theorems from Strassen's theory of asymptotic spectra. For a detailed description, we refer the reader to~\cite{strassen1988asymptotic,phd}.

A semiring $(S,+,\cdot,0,1)$ is a set $S$ equipped with a binary addition operation $+$, a~binary multiplication operation $\cdot$, and elements $0,1\in S$, such that for all $a,b,c\in S$, \changed{we have}
\begin{gather}
(a+b)+c=a+(b+c),\, a+b=b+a\\
0+a=a,\, 0\cdot a=0,\, 1\cdot a=a\\
(a\cdot b)\cdot c=a\cdot(b\cdot c)\\ 
a\cdot(b+c)=a\cdot b+a\cdot c.
\end{gather}
A semiring $(S, +, \cdot, 0, 1)$ is commutative if for all $a,b \in S$, \changed{we have} $a\cdot b=b\cdot a$.
For any natural number~$n\in\N$, let $n \in S$ denote the sum of $n$ times the element $1 \in S$. %

A preorder $\leq$ on $S$ is a relation such that for any $a,b,c\in S$, \changed{we have that} $a\leq a$, and that if $a\leq b$ and $b\leq c$, then $a\leq c$.
\begin{definition}\label{strassen preorder}
A preorder $\leq$ on $S$ is a Strassen preorder if for all $a,b,c \in S$ and $n,m \in \N$, \changed{we have}
\begin{gather}
\textnormal{$n\leq m$ in $\N$ if and only if $n\leq m$ in $S$}\\
\textnormal{if $a\leq b$, then $a+c\leq b+c$ and $a\cdot c\leq b\cdot c$}\\
\textnormal{if $b\neq 0$, then there exists an $r\in\N$ such that $a\leq r\cdot b$}.
\end{gather}
\end{definition}

Let $S = (S, +, \cdot, 0, 1)$ and $S' = (S', +, \cdot, 0, 1)$ be semirings.
A semiring homomorphism from $S$ to~$S'$ is a map $\phi: S\to S'$ such that $\phi(a + b) = \phi(a) + \phi(b)$, $\phi(a \cdot b) = \phi(a) \cdot \phi(b)$ for all~$a,b \in S$, and~$\phi(1) = 1$. %
Let $\R_{\geq 0} = (\R_{\geq0}, +, \cdot, 0, 1)$ be the semiring of non-negative real numbers with the usual addition and multiplication operations. The asymptotic spectrum $\bX(S, \leq)$ of the semiring $S = (S, +, \cdot, 0, 1)$ with respect to the preorder $\leq$ is the set of $\leq$-monotone semiring homomorphisms from $S$ to $\R_{\geq 0}$, i.e.
\begin{equation}\label{Eq: asymptotic spectrum}
\bX(S,\leq)\coloneqq\{\phi\in\Hom(S,\R_{\geq 0}) : \forall a,b\in S,~a\leq b~\Rightarrow~\phi(a)\leq\phi(b)\}.
\end{equation}

Let $a\in S$. The subrank of $a$ is defined as $\subrank(a)\coloneqq\max\{n\in\N:n\leq a\}$. The rank of $a$ is defined as $\rank(a)\coloneqq\min\{n\in\N:a\leq n\}$.  The asymptotic subrank and asymptotic rank of~$a$ are defined as
\begin{equation}\label{defasympsubrank}
\asympsubrank(a) \coloneqq \lim_{N\to\infty}\sqrt[N]{\subrank(a^N)},~\text{and}~\asymprank(a) \coloneqq \lim_{N\to\infty}\sqrt[N]{\rank(a^N)}.
\end{equation}
Fekete's lemma implies that the limits in \eqref{defasympsubrank} indeed exist and can be replaced by a supremum and an infimum, that is,
\[
\asympsubrank(a) = \sup_{N}\sqrt[N]{\subrank(a^N)},~\text{and}~\asymprank(a) = \inf_{N}\sqrt[N]{\rank(a^N)}.
\]
Strassen proved the following dual characterizations of $\asympsubrank(a)$ and $\asymprank(a)$ in terms of the asymptotic spectrum.
\begin{theorem}[{\cite[Theorem 3.8]{strassen1988asymptotic} \changed{and} \cite[Cor.~2.14]{phd}}]\label{thm: asymptotic subrank}
Let $S$ be a commutative semiring and let~$\leq$ be a Strassen preorder on $S$. For any $a\in S$ such that $1\leq a$ and $2\leq a^k$ for some $k\in\N$, \changed{we have that}
\begin{equation}
\asympsubrank(a)=\min_{\phi\in\bX(S,\leq)}\phi(a),~\text{and}~\asymprank(a) = \max_{\phi\in\bX(S,\leq)}\phi(a).
\end{equation}
\end{theorem}

Besides asymptotic subrank and rank, the asymptotic spectrum of a commutative semiring with respect to a Strassen preorder $\leq$ also characterizes the asymptotic preorder~$\asympleq$ associated to $\leq$. The asymptotic preorder $\asympleq$ associated to $\leq$ is defined by $a \asympleq b$ if there is a sequence of natural numbers $(x_n)_{n \in \N} \subseteq \N$ such that $\inf_n (x_n)^{1/n} = 1$ and such that for all $n \in \N$ \changed{we have} $a^n \leq x_n\cdot b^n$. The dual characterization is that $a\asympleq b$ if and only if for all $\phi \in \bX(S,\leq)$ \changed{it holds that} $\phi(a) \leq \phi(b)$.
 See \cite[Cor.~2.6]{strassen1988asymptotic} and see also \cite[Theorem 2.12]{phd}.

Finally, we mention that the asymptotic spectrum is well-behaved with respect to subsemirings. Let $S$ be a commutative semiring, let $\leq$ be a Strassen preorder on $S$, and let $T\subseteq S$ be a subsemiring, which means that $0,1 \in T$ and that $T$ is closed under addition and multiplication. Then clearly the restriction ${\leq}_T$ of $\leq$ to $T$ is a Strassen preorder on $T$. For any $\phi \in \bX(S, \leq)$ the restricted function~$\phi|_T$ is clearly an element of $\bX(T, {\leq}_T)$. The opposite is also true.

\begin{theorem}[{\cite[Cor.~2.7]{strassen1988asymptotic} \changed{and} \cite[Cor.~2.17]{phd}}]\label{restr}
Let $S$ be a commutative semiring, let~$\leq$ be a Strassen preorder on $S$, and let $T\subseteq S$ be a subsemiring. For every element~$\phi \in \bX(T, {\leq}|_T)$ there is an element~$\psi \in \bX(S, \leq)$ such that $\psi$ restricted to $T$ equals $\phi$.
\end{theorem}

We note that the proof of Theorem~\ref{restr} is nonconstructive.

\subsection{Semiring of graphs and the dual characterization of Shannon capacity}\label{known}
Let~$\cG$ be the set of isomorphism classes of (finite simple) graphs. The cohomomorphism preorder~$\leq$ on~$\cG$ is defined by $G\leq H$ if and only if $\overline{G}\to\overline{H}$, that is,~there is a graph homomorphism from the complement of~$G$ to the complement of $H$. Zuiddam proved in~\cite{zuiddam2018asymptotic} that $\cG = (\cG,\sqcup,\boxtimes,K_0,K_1)$ is a commutative semiring and that the cohomomorphism preorder $\leq$ is a Strassen preorder on $\cG$. 
By definition, the asymptotic spectrum of graphs~$\bX(\cG,\leq)$ consists of all maps $\phi : \cG \to \R_{\geq 0}$ such that, for all $G,H\in\cG$, \changed{the following conditions hold:}
\begin{gather}
\phi(G\sqcup H)=\phi(G)+\phi(H)\\
\phi(G\boxtimes H)=\phi(G)\cdot\phi(H)\\
\phi(\overline{K_1})=1\\
G\leq H \Rightarrow \phi(G)\leq\phi(H).
\end{gather}
Note that the subrank of a graph $G$ equals the independence number of $G$, since equation~\eqref{eq: independence number} is exactly 
\[
\alpha(G)=\max\{n\in\N:\overline{K_n}\leq G\}.
\]
By Theorem~\ref{thm: asymptotic subrank}, the Shannon capacity is dually characterized as
\begin{equation}\label{Eq: shannon capacity asymptotic spectrum}
\Theta(G)=\min_{\phi\in\bX(\cG,\leq)}\phi(G).
\end{equation}
The known elements belonging to the asymptotic spectrum of graphs are: the Lov\'asz theta function~$\vartheta$~\cite{lovasz1979shannon}, the fractional Haemers bound $\cH_{\smash{f}}^\F$ over any field $\F$~\cite{bukh2018fractional,blasiak2013graph}, the complement of projective rank $\overline{\xi}_f$~\cite{manvcinska2016quantum,cubitt2014bounds} and the fractional clique cover number $\overline{\chi}_f$~\cite{schrijver2003combinatorial}. Note that there are infinitely many elements in $\bX(\cG,\leq)$, due to the separation result in~\cite{bukh2018fractional} of the fractional Haemers bound over different fields. 
\begin{remark}\label{remark: largest elements}
We note that the fractional clique cover number is the pointwise largest element in~$\bX(\cG,\leq)$. \changed{See,~e.g.,~\cite{zuiddam2018asymptotic}.} This is because the rank of a graph $G$ equals the clique cover number and the asymptotic clique cover number equals the fractional clique cover number. 
\end{remark}

\subsection{Quantum variants of graph homomorphism}
We present mathematical definitions of the two quantum variants of graph homomorphisms, arising from the theory of non-local games and from quantum zero-error information theory, respectively.

\subsubsection{Quantum homomorphism}
\begin{definition}[Quantum homomorphism~\cite{manvcinska2016quantum}]\label{def: quantum homomorphism}
Let $G$ and $H$ be graphs. We say there is a quantum homomorphism from $G$ to $H$, and write $G\Qto H$, if there exist $d\in \N$ and $d \times d$ projectors~$E_g^h\in M(d,\C)$ for every $g\in V(G)$ and $h\in V(H)$, such that the following two conditions hold:
\begin{gather}
\textnormal{for every $g\in V(G)$ we have $\textstyle\sum_{h\in V(H)} E_{g}^h=I_d$}\\
\textnormal{if $ \{g,g'\}\in E(G)$ and $\{h,h'\}\not\in E(H)$, then $E_g^h E_{g'}^{h'}=0$.}
\end{gather}
\end{definition}
\begin{remark}\label{remquant}\hfill
\begin{itemize}
\item The first condition implies $E_g^h E_g^{h'}=0$ for all $g\in V(G)$ and $h\neq h'\in V(H)$. Namely, for a fixed $g\in V(G)$ and an arbitrary $h'\in V(H)$, $\textstyle\sum_{h\in V(H)} E_{g}^h=I_d$ implies $\textstyle\sum_{h\in V(H)} E_{g}^hE_{g}^{h'}=E_g^{h'}$. Since every $E_g^h$ is a projector, we have $\textstyle\sum_{h\neq h'} E_{g}^hE_g^{h'}=0$. We conclude that $E_{g}^hE_g^{h'}=0$ since projectors are also positive semidefinite.
\item For every collection of complex projectors $(E_g^h\in M(d,\C):g\in V(G), h\in V(H))$ satisfying the above two conditions, there exists a collection of real projectors which also satisfies the above two conditions. Namely, we take the collection of real matrices 
\[
(F_g^h=\begin{bmatrix}\Re(E_g^h)&\Im(E_g^h)\\ -\Im(E_g^h)&\Re(E_g^h)\\
\end{bmatrix}\in M(2d,\R):g\in V(G), h\in V(H)),
\]
where $\Re(E_g^h)$ and $\Im(E_g^h)$ denote the real part and the image part of $E_g^h$, respectively.
Noting that $E_g^hE_g^h=(\Re(E_g^h)^2-\Im(E_g^h)^2)+i(\Re(E_g^h)\Im(E_g^h)+\Im(E_g^h)\Re(E_g^h))=(\Re(E_g^h)+i\Im(E_g^h))=E_g^h$, we have 
\begin{align*}
F_g^hF_g^h &= \begin{bmatrix}\Re(E_g^h)^2-\Im(E_g^h)^2&\Re(E_g^h)\Im(E_g^h)+\Im(E_g^h)\Re(E_g^h)\\ -\Re(E_g^h)\Im(E_g^h)-\Im(E_g^h)\Re(E_g^h)&\Re(E_g^h)^2-\Im(E_g^h)^2 \end{bmatrix}\\
&=\begin{bmatrix}\Re(E_g^h)&\Im(E_g^h)\\ -\Im(E_g^h)&\Re(E_g^h)\\
\end{bmatrix}=F_g^h.
\end{align*}
Moreover, it is easy to verify that $(F_g^h:g\in V(G), h\in V(H))$ satisfies the conditions in Definition~\ref{def: quantum homomorphism}~\cite{manvcinska2016quantum}.
\end{itemize}
\end{remark}
It is easy to see that $G\to H$ implies $G\Qto H$. The opposite direction is not true~\cite{manvcinska2016quantum}. The quantum cohomomorphism preorder on graphs is defined by $G\leq_q H$ if and only if $\overline{G}\Qto\overline{H}$, and the quantum independence number as $\alpha_q(G)\coloneqq\max\{n\in\N:\overline{K_n}\leq_q G\}$. The quantum Shannon capacity is defined as $\Theta_q(G) \coloneqq\lim_{N\to\infty}\sqrt[N]{\alpha_q(G^{\boxtimes N})}=\sup_{N}\sqrt[N]{\alpha_q(G^{\boxtimes N})}$.

\subsubsection{Entanglement-assisted homomorphism}\label{conf}
\begin{definition}[Entanglement-assisted homomorphism~\cite{cubitt2014bounds}]\label{def: entanglement-assisted homomorphism}
Let $G$ and $H$ be graphs. We say there is an entanglement-assisted homomorphism from $G$ to $H$, and write $G\EAto H$, if there exist $d\in \N$ and $d\times d$ positive semidefinite matrices $\rho$ and $(\rho_g^h\in M(d,\C): g\in V(G),\, h\in V(H))$, %
such that the following two conditions hold
\begin{gather}
\textnormal{for every $g\in V(G)$ we have $\textstyle\sum_{h\in V(H)} \rho_{g}^h=\rho$}\\
\textnormal{if $\{g,g'\}\in E(G)$ and $\{h,h'\}\not\in E(H)$, then $\rho_g^h \rho_{g'}^{h'}=0$}.
\end{gather}
\end{definition}
\begin{remark}
We note that the positive semidefinite matrix $\rho$ can be further restricted to be positive definite. 
\end{remark}
The entanglement-assisted cohomomorphism preorder is defined by~$G\leq_* H$ if and only if $\overline{G}\EAto\overline{H}$. The entanglement-assisted independence number can be defined as~$\alpha_*(G)=\max\{n\in\N:\overline{K_n}\leq_* G\}$. The entanglement-assisted Shannon capacity of $G$ is defined as~$\Theta_*(G)\coloneqq\lim_{N\to\infty}\sqrt[N]{\alpha_*(G^{\boxtimes N})}=\sup_{N}\sqrt[N]{\alpha_*(G^{\boxtimes N})}$. 

It is easy to see that $G\Qto H$ implies $G\EAto H$. It remains unknown whether the reverse direction is true. In fact, as pointed out in~\cite{manvcinska2016quantum}, $G\Qto H$ can be interpreted in the zero-error communication setting by restricting to the use of maximally entanglement state and projective measurements. 

\subsection{(Entanglement-assisted) zero-error capacity of quantum channels}\label{sec: ea nc graph}
For related definitions in quantum information theory, we refer the reader to~\cite{Nielsen2010}. We use $A$ and~$B$ to denote the (finite-dimensional) Hilbert spaces of the sender (Alice) and the receiver (Bob), respectively. Let $\cL(A,B)$ be the space of linear operators from $A$ to $B$. Let $\cL(A) \coloneqq \cL(A,A)$. The space $\cL(A)$ is isomorphic to the matrix space $M(n,\C)$ with $n=\dim(A)$. Let $\cD(A)\subseteq \cL(A)$ be the set of all (mixed) quantum states, i.e.~all trace-$1$ positive semidefinite operators in $\cL(A)$. A quantum state $\rho\in\cD(A)$ is pure if it has rank $1$, i.e.~if it can be written as $\rho=\Proj{\psi}$ for some unit vector $\ket{\psi}\in A$. 
The support of %
a positive semidefinite matrix $P\in\cL(A)$ is the subspace of $A$ spanned by the eigenvectors with positive eigenvalues. A quantum channel $\cN:\cL(A)\to\cL(B)$ can be characterized by a completely positive and trace-preserving (CPTP) map. This is equivalent to saying that $\cN$ is of the form $\cN(\rho)=\sum_i N_i\rho N_i^\dagger$ for some linear operators $\{N_i\}_i\subseteq\cL(A,B)$, called the Choi--Kraus operators associated to~$\cN$, satisfying $\sum_i N_i^\dagger N_i= I_A$.{\footnote{The Choi--Kraus operators are not unique, but if we have $\cN(\rho)=\sum_{i=1}^m N_i\rho N_i^\dagger=\sum_{i=1}^m \hat{N}_i\rho \hat{N}_i^\dagger$, then there exists an $m$-by-$m$ unitary matrix $U=[u_{i,j}]_{i,j=[m]}$ such that $N_i=\sum_{j=1}^m u_{i,j}\hat{N}_j$ for each $i\in[m]$. See, e.g.,~\cite[Theorem~8.2]{Nielsen2010}.}}  

We focus on the setting in which Alice and Bob use a quantum channel to transmit classical zero-error messages.
To transmit $k$ classical messages to Bob through the quantum channel $\cN$, Alice prepares $k$ pairwise orthogonal states $\rho_1,\dots,\rho_k \in \cD(A)$, where orthogonality is defined with respect to the Hilbert--Schmidt inner product $\langle \rho, \sigma \rangle = \Tr(\rho^\dagger\sigma)$. Bob needs to distinguish the output states $\cN(\rho_1),\dots,\cN(\rho_k)$ perfectly, in order to obtain the messages without error. This is only possible when the output states are pairwise orthogonal. In this situation, without loss of generality, Alice may select the $\rho_i=\Proj{\psi_i}$ to be pure states for all~$i\in[k]$. Note that $\cN(\Proj{\psi})\perp \cN(\Proj{\phi})$ if and only if $\ket{\psi}\!\bra{\phi}\perp N_i^\dagger N_j$ for all $i\neq j$, where $\{N_i\}_i$ are the Choi--Kraus operators of $\cN$. 
We now see that the number of messages one can transmit through the channel $\cN$ is  determined by the linear space of matrices $S=\linspan\{N_i^\dagger N_j\}_{i,j}\subseteq\cL(A)$. 
Duan, Severini and Winter called the linear space $S$ the non-commutative graph (nc-graph) of a quantum channel~$\cN$~\cite{duan2013}. The nc-graphs may be thought of as the quantum generalization of confusability graphs of classical channels, mentioned in Section~\ref{conf}. In this analogy, for an nc-graph $S\subseteq \cL(A)$ the density operators~$\rho,\sigma\in\cD(A)$ are input symbols of the channel, and they are ``non-adjacent'' in the nc-graph~$S$ if~$\ket{\phi}\!\bra{\psi}\perp S$ for all $\ket{\phi}$ and~$\ket{\psi}$ in the support of $\rho$ and $\sigma$, respectively. As in the classical setting, ``non-adjacent vertices'' are nonconfusable.

Note that for every quantum channel $\cN$, the associated nc-graph $S$ satisfies $S^\dagger=S$ and $I_A\in S$, where $I_A$ is the identity operator in $\cL(A)$. It is shown in~\cite{duan2009super,cubitt2011} that any subspace $S\subseteq\cL(A)$ that satisfies $S^\dagger = S$ and $I_A \in S$ is associated to some quantum channel.
From now, we define a non-commutative graph or nc-graph as a subspace $S\subseteq\cL(A)$ satisfying $S^\dagger=S$ and $I_A\in S$. 
We define the independence number $\alpha(S)$ as the maximum $k$ such that there exist pure states $\ket{\psi_1},\dots,\ket{\psi_k}$ satisfying $\ket{\psi_i}\!\bra{\psi_j}\perp S$ for all $i\neq j\in[k]$. 
The Shannon capacity is defined as $\Theta(S):=\lim_{N\to\infty}\sqrt[N]{\alpha(S^{\otimes N})}$, where $S_1\otimes S_2:=\linspan\{E_1\otimes E_2:E_1\in S_1,E_2\in S_2\}$ denotes the tensor product of $S_1$ and $S_2$. One verifies that if $S_1$ and $S_2$ are the nc-graphs of $\cN_1$ and~$\cN_2$ respectively, then the tensor product $S_1\otimes S_2$ is the nc-graph of the quantum channel $\cN_1\otimes\cN_2$.
Then (the logarithm of)~$\Theta(S)$ is exactly the classical zero-error capacity of quantum channels whose nc-graph is $S$~\cite{duan2013}.

In the quantum setting, it will be natural to consider that Alice and Bob are allowed to share entanglement to assist the information transmission. To make use of the entanglement, say {$\ket{\Omega}\in A_0\otimes B_0$}, Alice prepares $k$ quantum channels $\cE_1,\dots,\cE_k:\cL(A_0)\to\cL(A)$ for encoding the classical messages. To send the $i$th message, Alice applies $\cE_i$ to her part of $\ket{\Omega}$, and sends the output state to Bob via the quantum channel~$\cN$. Bob needs to perfectly distinguish the output states $\rho_i=((\cN\circ\cE_i)\otimes\cI_{B_0})(\ket{\Omega}\!\bra{\Omega})$ for $i\in[k]$. The following lemma shows that the maximum number of classical message which can be sent via the quantum channel~$\cN$ in the presence of entanglement can be also characterized by the nc-graph $S$, as also mentioned in~\cite{duan2013,stahlke2016}. 
\begin{lemma}[\cite{stahlke2016}]\label{claim: orthogonal}
Let $\cN$, $S$, $\cE_i$ and $\ket{\Omega}$ as above.
Let $\{E_{i,\ell}\}_\ell$ and $\{E_{j,\ell'}\}_{\ell'}$ be the Choi--Kraus operators of $\cE_i$ and~$\cE_j$, respectively, for $i\neq j\in[k]$. Then 
\[
((\cN\circ\cE_i)\otimes\cI_{B_0})(\ket{\Omega}\!\bra{\Omega})\perp((\cN\circ\cE_j)\otimes\cI_{B_0})(\ket{\Omega}\!\bra{\Omega})
\]
is equivalent to 
\[
\linspan\{E_{i,\ell}\Tr_{B_0}(\ket{\Omega}\!\bra{\Omega})E^\dagger_{j,\ell'}\}\perp S.
\]
\end{lemma}
\begin{proof}
\changed{The first condition is equivalent to 
$$\bra{\Omega}(E_{i,\ell}^\dagger N_m\otimes I)(N_nE_{j,\ell'}\otimes I)\ket{\Omega}=0\quad\forall\ \ell,\ell',m,n.$$
This is further equivalent to
$$\Tr\bigl(N_nE_{j,\ell'}\Tr_{B_0}(\ket{\Omega}\!\bra{\Omega})E_{i,\ell}^\dagger N_m^\dagger \bigr)=0\quad\forall\ \ell,\ell',m,n,$$
which concludes the proof.}
\end{proof}

We call $(\ket{\Omega},\{\cE_1,\dots,\cE_k\})$ a size-$k$ entanglement-assisted independent set of %
$S$. 
Let $\alpha_*(S)$ be the maximum size of an entanglement-assisted independent set of $S$. 
The entanglement-assisted Shannon capacity of $S$ is defined as %
$\Theta_{*}(S)\coloneqq\lim_{N\to\infty}\sqrt[N]{\alpha_*(S^{\otimes N})}$.   %

\subsection{Semiring of non-commutative graphs and preorders}
Recall that an nc-graph is a subspace $S \subseteq \cL(A)$ satisfying $S^\dagger = S$ and $I_A \in S$.
We point out that every classical graph $G$ naturally corresponds to an nc-graph $S_G$. 
Namely, for any graph $G$, let~$\{\ket{g}:g\in G\}$ be the standard orthonormal basis of $\C^{|V(G)|}$ and define
\[
S_G \coloneqq \linspan\{\ket{g}\!\bra{g'}:g=g'\in V(G)~{\rm or}~\{g,g'\}\in E(G)\} \subseteq \cL(\C^{|V(G)|}).
\]
For the nc-graphs corresponding to the complement of the complete graphs we use the notation
\[
\overline{\cK_n} \coloneqq S_{\overline{K_n}}=\linspan\{\Proj{i}:i\in[n]\}\subseteq\cL(\C^n).
\]
It is worth noting that $\overline{\cK_n}$ is the nc-graph of the $n$-message noiseless classical channel, which maps $\ket{m}\!\bra{m'}$ to $\delta_{m,m'}\Proj{m}$ for all $m,m'\in[n]$. %

We say two nc-graphs $S_1$ and $S_2$ are isomorphic if they are equal up to a unitary transformation, i.e.\ if~$S_2=U^\dagger S_1 U$ for some unitary matrix $U$. Let $\cS$ be the set of isomorphism classes of nc-graphs. Analogous to the operations in the semiring of graphs, for two nc-graphs $S_1\subseteq \cL(A_1)$ and $S_2\subseteq\cL(A_2)$, the ``disjoint union'' is their direct sum $S_1\oplus S_2\subseteq\cL(A_1)\oplus \cL(A_2) \subseteq \cL(A_1 \oplus A_2)${\footnote{More precisely, $S_1\oplus S_2=\linspan\{\begin{bmatrix}X_1 &0\\0&0\end{bmatrix}, \begin{bmatrix}0&0\\0&X_2 \end{bmatrix}:X_1\in S_1,X_2\in S_2\}$.}} and the ``strong graph product'' is their tensor product $S_1\otimes S_2\subseteq \cL(A_1)\otimes\cL(A_2)\cong\cL(A_1\otimes A_2)$. The reader readily verifies the following.

\begin{theorem}\label{thm: semiring of nc-graphs}
$\cS = (\cS,\oplus,\otimes,\overline{\cK_0},\overline{\cK_1})$ is a commutative semiring. 
\end{theorem}

\begin{lemma}\label{injhom}
The map $\iota: \cG \to \cS : G \mapsto S_G$ is an injective semiring homomorphism.
\end{lemma}
{ \begin{proof}
It is easy to see that $\iota$ is a semiring homomorphism by verifying that $S_{K_0}=0$, $S_{K_1}=\C$, $S_{G\sqcup H}=S_G\oplus S_H$ and $S_{G\boxtimes H}=S_G\otimes S_H$. To see that $\iota$ is an injection, we need to prove that if $S_G$ and $S_H$ are (unitary) isomorphic (as nc-graphs), then $G$ and $H$ are (permutation) isomorphic (as graphs). This has been verified in~\cite[Proposition 3.1]{ORTIZ2015128}.
\end{proof}
}

We can adapt the definitions in~\cite{stahlke2016} to obtain the cohomomorphism preorder and the entanglement-assisted cohomomorphism preorder on nc-graphs as follows.
\begin{definition}\label{def: homomorphic preorder of quantum channels}
The cohomomorphism preorder $\leq$ is defined on $\cS$ by, for any nc-graphs $S_1\subseteq\cL(A_1)$ and~$S_2\subseteq\cL(A_2)$, letting
$S_1\leq S_2$ if there exists $E=\linspan\{E_i\}_i\subseteq\cL(A_1,A_2)$ satisfying $\sum_{i} E_i^\dagger E_i= I_{A_1}$, such that $E S_1^\perp E^\dagger \perp S_2$, where $S_1^\perp\coloneqq\{X\in\cL(A_1):\forall Y\in S_1\, \Tr(X^\dagger Y)=0\}$.
\end{definition}
\begin{definition}
The entanglement-assisted cohomomorphism preorder $\leq_*$ is defined on $\cS$ by, for any nc-graphs $S_1\subseteq\cL(A_1)$ and $S_2\subseteq\cL(A_2)$, letting $S_1\leq_* S_2$ if there exist a positive definite~$\rho\in\cD(A_0)$ and $E=\linspan\{E_i\}_i\subseteq\cL(A_1\otimes\rho,A_2)$ satisfying $\sum_{i} E_i^\dagger E_i= I_{A_1 \otimes A_0}$, such that $E (S_1^\perp\otimes\rho) E^\dagger\perp S_2$.
\end{definition}

\begin{lemma}\label{preorderrel}
If $S\leq T$, then $S\leq_*T$.
\end{lemma}
\begin{proof}
Take the positive definite matrix $\rho$ in the definition of $S\leq_* T$ to be the element $1$.
\end{proof}

The idea behind the above definitions is as follows. Recall that $G\leq H$ if there exists a graph homomorphism from $\overline{G}$ to $\overline{H}$. In other words, there exists a vertex map $f:V(G)\to V(H)$ which maps non-adjacent vertices to non-adjacent vertices. %
Since we may view quantum states as vertices and matrices in the nc-graph as edges in nc-graphs, it is natural to adapt the ``vertex map'' among nc-graphs $S_1\subseteq\cL(A_1)$ and $S_2\subseteq\cL(A_2)$ as a CPTP map $\cE:\cL(A_1)\to\cL(A_2)$, specified by the Choi--Kraus operators $\{E_i\}_i\subseteq\cL(A_1,A_2)$. Now for ``non-adjacent vertices''  $\Proj{\psi}$ and $\Proj{\phi}$ in $S_1$, we require $\cE(\Proj{\psi})$ and $\cE(\Proj{\phi})$ are ``non-adjacent'' in $S_2$. The former is equivalent to $\ket{\psi}\!\bra{\phi}\perp S_1$ and the latter is equivalent to $E_i\ket{\psi}\!\bra{\phi}E_j^\dagger\perp S_2$ for all $i,j$. The definition of $S_1\leq S_2$ is then obtained naturally. 

To see that the above definitions are meaningful, Stahlke in~\cite{stahlke2016} also points out the following.
\begin{lemma}\label{prop: nc-graph independent numbers}
Let $S$ be an nc-graph. Then
\begin{enumerate}[label=\upshape(\roman*)]
\item $\alpha(S)=\max\{n\in\N:\overline{\cK_n}\leq S\}$
\item $\alpha_*(S)=\max\{n\in\N:\overline{\cK_n}\leq_* S\}$.
\end{enumerate}
\end{lemma}
\begin{proof}
We provide a detailed proof in Appendix~\ref{sec: proof of 2.8}.
\end{proof}

\section{Dual characterization of entanglement-assisted Shannon capacity and quantum Shannon capacity of graphs}\label{sec: dual characterization graphs}

In this section, we first prove that the entanglement-assisted capacity $\Theta_*$ and the quantum Shannon capacity $\Theta_q(\cdot)$ can be characterized by applying Strassen's theory of asymptotic spectra, and present elements in the corresponding asymptotic spectra. We also discuss the relations between two important conjectures in quantum zero-error information theory.

\subsection{Entanglement-assisted Shannon capacity $\Theta_*(G)$ of a graph}
We first prove that the entanglement-assisted cohomomorphism preorder $\leq_*$ (Definition~\ref{def: entanglement-assisted homomorphism}) is a Strassen preorder on the semiring of graphs $\cG$.
\begin{lemma}\label{lemma: EA homo preorder strassen}
For any graphs $G, H, K, L$ and any $n,m \in \N$, we have
\begin{enumerate}[label=\upshape(\roman*)]
\item $G\leq_* G$
\item if $G\leq_* H$ and $H\leq_* L$, then $G\leq_* L$
\item $\overline{K_m} \leq_* \overline{K_n}$ if and only if $m \leq n$
\item if $G \leq_* H$ and $K\leq_* L$ then $G\sqcup K \leq_* H\sqcup L$ and $G \boxtimes K \leq_* H\boxtimes L$
\item if $H\neq \overline{K_0}$, then there is an $r \in \N$ with $G \leq_* \overline{K_r} \boxtimes H$.
\end{enumerate}
\end{lemma}
\begin{proof}
(i) %
In general, if $G\leq H$, then $G \leq_* H$.
It is clear that $G \leq G$. Therefore also $G \leq_* G$.

(ii) We adapt the proof of~\cite[Lemma 2.5]{manvcinska2016quantum} to show transitivity.
Assume $G \leq_* H$ and $H \leq_* L$.
Let $\rho$, $(\rho_g^h: g \in V(G), h \in V(H))$ and $\sigma$, $(\sigma_h^\ell : h \in V(H), c \in V(L))$ be corresponding positive semidefinite matrices, as in Definition~\ref{def: entanglement-assisted homomorphism}. For $g \in V(G), l \in V(L)$, let
$\tau_g^\ell = \sum_{h \in V(H)} \rho_g^h \otimes \sigma_h^\ell$. Note that $\tau_g^\ell$ is  positive semidefinite for all $g\in V(G)$ and $l\in V(L)$. 
We have
\begin{equation}\label{eq: sum of rho g h}
\sum_{\ell\in V(L)} \tau_g^\ell=\sum_{\substack{h\in V(H)\\ \ell\in V(L)}} \rho_g^h\otimes \sigma_h^\ell= \sum_{h\in V(H)} \rho_g^h \otimes \sum_{\ell\in V(L)} \sigma_h^\ell=\rho \otimes \sigma.
\end{equation}
For all $\{g,g'\}\not\in E(G)$ and $\{\ell,\ell'\}\in E(L)$ or $\ell=\ell'$, we have
\begin{equation}\label{eq: non-adjacent}
\rho_g^\ell \rho_{g'}^{\ell'}=\sum_{h,h' \in V(H)} \rho_g^h\rho_{g'}^{h'} \otimes \rho_h^\ell\rho_{h'}^{\ell'}= \sum_{\{h,h'\}\not\in E(H)} \rho_g^h\rho_{g'}^{h'} \otimes \rho_h^\ell\rho_{h'}^{\ell'}=0
\end{equation}
 where the second equality holds since $\rho_g^h \rho_{g'}^{h'}=0$ for all $\{g,g'\}\not\in E(G)$ and $\{h,h'\}\in E(H)$ or $h=h'$, and the third equality holds since $\rho_h^\ell\rho_{h'}^{\ell'}=0$ for all $\{h,h'\}\not\in E(H)$ and $\{\ell,\ell'\}\in E(L)$ or $\ell=\ell'$. We conclude $G\leq_* L$.

(iii) %
We know that $m \leq n$ implies $\overline{K_m} \leq \overline{K_n}$, and thus $\overline{K_m} \leq_* \overline{K_n}$. To see that $\overline{K_m} \leq_* \overline{K_n}$ implies $m\leq n$, we note that $G\leq_* H$ implies $\vartheta(G)\leq\vartheta(H)$~\changed{\cite{cubitt2014bounds}}.
Thus $\overline{K_m} \leq_* \overline{K_n}$ implies $m=\vartheta(\overline{K_m})\leq\vartheta(\overline{K_n})=n$.

(iv) 
Assume that $G\leq_* H$ and $K\leq_* L$.
Let $\rho, (\rho_g^h:g\in V(G), h\in V(H))$ and $\sigma, (\sigma_k^\ell:k\in V(K), \ell\in V(L))$ be corresponding positive semidefinite matrices, as in Definition~\ref{def: entanglement-assisted homomorphism}. Let
\[
\tau_u^v = \begin{cases} \rho_u^v \otimes \sigma & \textnormal{if $u \in V(G), v\in V(H)$}\\
\rho \otimes \sigma_u^v & \textnormal{if $u \in V(K), v \in V(L)$}\\
0 & \textnormal{otherwise.}
\end{cases}
\]
One readily verifies that $\tau_u^v$ is positive semidefinite for all $u\in V(G\sqcup K)$ and $v\in V(H\sqcup L)$. Moreover, for every $u\in V(G)$ we have $\sum_{v \in V(H\sqcup L)} \tau_u^v=\sum_{v \in V(H)} \tau_u^v+\sum_{v \in V(L)} \tau_u^v= \rho \otimes \sigma$, and for every~$u \in V(K)$ we have $\sum_{v \in V(H\sqcup L)} \tau_u^v=\sum_{v \in V(H)} \tau_u^v+\sum_{v \in V(L)} \tau_u^v= \rho \otimes \sigma$. One verifies directly that $\tau_u^v\tau_{u'}^{v'}=0$ for all $\{u,u'\}\in E(\overline{G\sqcup K})$ and $\{v,v'\} \not\in E(\overline{H\sqcup L})$. We conclude $G\sqcup K \leq_* H\sqcup L$.

To prove that $G \boxtimes K \leq_* H\boxtimes L$, let $\tau_{(g,k)}^{(h,\ell)} = \rho_g^h \otimes \sigma_k^\ell$ for all $g,h,k,\ell$. One readily verifies that these operators satisfy the required conditions.

(v) %
For the cohomomorphism preorder it is not hard to see that for all $G,H\neq \overline{K_0}$, there is an~$r \in \N$ with $G \leq \overline{K_r} \boxtimes H$~\cite[Lemma 4.2]{zuiddam2018asymptotic}. Therefore, $G \leq_* \overline{K_r} \boxtimes H$.
\end{proof}

Let $\bX(\cG,\leq_*)$
be the asymptotic spectrum of graphs with respect to the entanglement-assisted cohomomorphism preorder $\leq_*$, i.e.
\begin{equation}
\bX(\cG,\leq_*)=\{\phi\in\Hom(\cG,\R_{\geq 0}):\forall G,H\in \cG,~G\leq_* H~\Rightarrow~\phi(G)\leq\phi(H)\}.
\end{equation}\label{eq: X(G,leq_*)}
Together with Theorem~\ref{thm: asymptotic subrank}, we obtain the following dual characterization of the entanglement-assisted Shannon capacity of graphs, $\Theta_*(G)$.
\begin{theorem} Let $G$ be a graph. Then
\[
\Theta_*(G)=\min_{\phi\in\bX(\cG,\leq_*)}\phi(G).
\]
\end{theorem}

Since $G\leq H$ implies $G\leq_* H$, we have $\bX(\cG,\leq_*)\subseteq\bX(\cG,\leq)$. As we mentioned already in the proof of Lemma~\ref{lemma: EA homo preorder strassen}, the Lov\'asz theta function is $\leq_*$-monotone~\changed{\cite{cubitt2014bounds}}.
This implies the following.

\begin{theorem}\label{thm: elements in leq_*}
$\vartheta\in\bX(\cG,\leq_*)$.
\end{theorem}

In fact, the following conjecture has been mentioned in~\cite{Beigi2010,5961832,Leung2012,duan2013,6466384,cubitt2014bounds}.
\begin{conjecture}\label{conj: theta}
$\Theta_*(G)=\vartheta(G)$ for all graphs $G$.
\end{conjecture}
It would be interesting to show that this conjecture is true by proving that $\vartheta$ is the minimal element in $\bX(\cG,\leq_*)$, or even the only point in $\bX(\cG,\leq_*)$.

\subsection{Quantum Shannon capacity}
We begin by proving that the quantum cohomomorphism preorder $\leq_q$ (Definition~\ref{def: quantum homomorphism}) is a Strassen preorder on the semiring of graphs.
\begin{lemma}\label{lemma: quantum homo preorder strassen}
For any graphs $G, H, K, L$ and any $n,m \in \N$, we have
\begin{enumerate}[label=\upshape(\roman*)]
\item $G\leq_q G$
\item if $G\leq_q H$ and $H\leq_q L$, then $G\leq_q L$
\item $\overline{K_m} \leq_q \overline{K_n}$ if and only if $m \leq n$
\item if $G \leq_q H$ and $K\leq_q L$ then $G\sqcup K \leq_q H\sqcup L$ and $G \boxtimes K \leq_q H\boxtimes L$
\item if $H\neq \overline{K_0}$, then there is an $r \in \N$ with $G \leq_q \overline{K_r} \boxtimes H$.
\end{enumerate}
\end{lemma}
\begin{proof}
(i) 
In general, if $G\leq H$, then $G \leq_q H$.
It is clear that $G \leq G$. Therefore also $G \leq_q G$.

(ii) Quantum homomorphisms are known to be transitive in the sense that if $G \Qto H$ and $H \Qto L$, then $G \Qto L$~\cite[Lemma 2.5]{manvcinska2016quantum}. Therefore, if $G \leq_q H$ and $H \leq_q L$, then $G \leq_q L$.

(iii) It is known that $K_m\Qto K_n$ if and only if $m\leq n$~\cite[Lemma 2.6]{manvcinska2016quantum}. Thus $\overline{K_m} \leq_q \overline{K_n}$ if and only if $m \leq n$.

(iv) 
Assume $G\leq_q H$ and $K\leq_q L$. Let $(E_g^h : g\in V(G), h\in V(H))$ and $(F_k^\ell:k\in V(K), \ell\in V(L))$ be the corresponding collections of projectors, as in Definition~\ref{def: quantum homomorphism}.  To prove $G\sqcup K \leq_q H\sqcup L$, let
\[
D_u^v = \begin{cases} E_u^v \otimes I & \textnormal{if $u \in V(G), v\in V(H)$}\\
I \otimes F_u^v & \textnormal{if $u \in V(K), v \in V(L)$}\\
0 & \textnormal{otherwise.}
\end{cases}
\]
It is easy to see that $D_u^v$ is a projector for every $u\in V(G\sqcup K)$ and $v\in V(H\sqcup L)$. Moreover, for every $u\in V(G)$ we have $\sum_{v \in V(H\sqcup L)} D_u^v=\sum_{v \in V(H)} D_u^v+\sum_{v \in V(L)} D_u^v= I \otimes I$, and for every $u \in V(K)$ we have $\sum_{v \in V(H\sqcup L)} D_u^v=\sum_{v \in V(H)} D_u^v+\sum_{v \in V(L)} D_u^v= I \otimes I$. It is also easy to verify that $D_u^vD_{u'}^{v'}=0$ for all $\{u,u'\}\in E(\overline{G\sqcup K})$ and $\{v,v'\} \not\in E(\overline{H\sqcup L})$. We conclude that $G\sqcup K \leq_q H\sqcup L$.

To prove $G \boxtimes K \leq_q H\boxtimes L$, let $D_{(g,k)}^{(h,\ell)} = E_g^h \otimes F_k^\ell$ for all $g,h,k,\ell$. One can also verify that these operators satisfy the required conditions to conclude $G \boxtimes K \leq_q H\boxtimes L$.

(v) For the cohomomorphism preorder it is not hard to see that for all $G,H\neq \overline{K_0}$, there is an~$r \in \N$ with $G \leq \overline{K_r} \boxtimes H$~\cite[Lemma 4.2]{zuiddam2018asymptotic}. Therefore, $G \leq_q \overline{K_r} \boxtimes H$.
\end{proof}
Let $\bX(\cG,\leq_q)$ be the asymptotic spectrum of graphs with respect to the quantum cohomomorphism preorder $\leq_q$, i.e.
\begin{equation}\label{eq: X(G,leq_q)}
\bX(\cG,\leq_q)=\{\phi\in\Hom(\cG,\R_{\geq 0}): \forall G,H\in \cG,~G\leq_q H~\Rightarrow~\phi(G)\leq\phi(H)\}.
\end{equation}
Together with Theorem~\ref{thm: asymptotic subrank}, we obtain the following dual characterization of the quantum Shannon capacity of graphs.

\begin{theorem}
Let $G$ be a graph. Then
\[
\Theta_q(G)=\min_{\phi\in\bX(\cG,\leq_q)}\phi(G).
\]
\end{theorem}

We know that if $G\leq H$, then $G\leq_q H$. It is also easy to see that $G\leq_q H$ implies $G\leq_* H$. Therefore, $\bX(\cG,\leq_*)\subseteq\bX(\cG,\leq_q)\subseteq\bX(\cG,\leq)$. 

\begin{theorem}\label{thm: element in leq q}
We have
\[
\{\vartheta, \cH_f^\C, \cH_f^\R, \overline{\xi}_f\}\subseteq\bX(\cG,\leq_q).
\]
\changed{Moreover,  $\overline{\chi}_f\not\in\bX(\cG,\leq_q)$ and for any odd prime $p$ such that there exists a Hadamard matrix of size~$4p$, we have $\cH_f^{\F_p} \not\in\bX(\cG,\leq_q)$. Therefore, we have a proper inclusion
\[
\bX(\cG,\leq_q)\subsetneq\bX(\cG,\leq).
\]}
\end{theorem}
\begin{proof}
We know that $\{\vartheta, \cH_{\smash{f}}^\C, \cH_{\smash{f}}^\R, \overline{\xi}_{\smash{f}}\} \subseteq  \bX(\cG,\leq)$, so to prove that $\{\vartheta, \cH_{\smash{f}}^\C, \cH_{\smash{f}}^\R, \overline{\xi}_{\smash{f}}\}\subseteq \bX(\cG,\leq_q)$, it remains to show that the functions $\vartheta, \cH_f^\C, \cH_f^\R, \overline{\xi}_f$ are $\leq_q$-monotone.

Man{\v{c}}inska and Roberson proved in \cite{manvcinska2016quantum} 
that the Lov\'asz theta function $\vartheta$ and the complement of projective rank $\overline{\xi}_f$ are $\leq_q$-monotone.
We prove that $\cH_f^\C$ is $\leq_q$-monotone. Suppose $G\leq_q H$. Let $E_g^h$ be corresponding complex $d'\times d'$ projector for all $g\in V(G)$ and $h\in V(H)$. 
Let 
\[
M=\sum_{h,h'\in V(H)}e_he_{h'}^\dagger\otimes M_{h,h'}\in M(|V(H)|,\C)\otimes M(d,\C)
\]
be a $d$-representation of $H$ over $\C$. We construct a $dd'$-representation of $G$ over $\C$ as follows. Let 
\[
M'=\sum_{g,g'\in V(G)}e_ge_{g'}^\dagger\otimes M'_{g,g'}\in M(|V(G)|,\C)\otimes M(dd',\C),
\] 
with
\[
M'_{g,g'}=\sum_{h,h'\in V(H)}M_{h,h'}\otimes E_g^hE_{g'}^{h'}\in M(dd',\C).
\]
To see that $M'$ is a $dd'$-representation of $G$, we first show $M'_{g,g}=I_{dd'}$ for all $g\in V(G)$. Note that
\[
M'_{g,g}=\sum_{h,h'\in V(H)} M_{h,h'}\otimes E_g^hE_{g}^{h'}=\sum_{h\in V(H)} M_{h,h}\otimes E_g^hE_g^h=I_d\otimes I_{d'},
\]
where the second equality uses $E_g^hE_{g}^{h'}=0$ for all $g\in V(G)$ and $h\neq h'$ (Remark~\ref{remquant}), and the last equality uses the facts that $M_{h,h}=I_d$ for all $h\in V(H)$ and $\sum_{h\in V(H)}  E_g^hE_g^h=\sum_{h\in V(H)}E_g^h=I_{d'}$. On the other hand, we show $M'_{g,g'}=0$ if $g\neq g'$ and $\{g,g'\}\not\in E(G)$. In this case, we have
\[
M'_{g,g'}=\sum_{h,h'\in V(H)} M_{h,h'}\otimes E_g^hE_{g'}^{h'}=\sum_{\{h,h'\}\not\in E(H)~and~h\neq h'} M_{h,h'}\otimes E_g^hE_{g'}^{h'}=0,
\]
where the second equality use the fact that $E_g^hE_{g'}^{h'}=0$ for all $\{g,g'\}\in E(\overline{G})$ and $\{h,h'\}\not\in E(\overline{H})$, and the last equality holds since $M_{h,h'}=0$ for all $h\neq h', \{h,h'\}\not\in E(H)$. %
Thus $M'$ is a $dd'$-representation of $G$ over $\C$. 

Next we prove that $ \rk(M')\leq d'\rk(M)$. %
We factor $M'$ as
\begin{align*}
M'&=\sum_{g,g'}e_ge^\dagger_{g'}\otimes\Bigl(\sum_{h,h'}M_{h,h'}\otimes E_g^hE_{g'}^{h'}\Bigr)\\
&=\Bigl(\sum_{g,h}e_ge^\dagger_{h}\otimes I_d\otimes E_g^h\Bigr)\Bigl(\sum_{h,h'}e_he^\dagger_{h'}\otimes M_{h,h'}\otimes I_{d'}\Bigr)\Bigl(\sum_{g',h'}e_{h'}e_{g'}^\dagger\otimes I_d\otimes E_{g'}^{h'}\Bigr)\\
&=\Bigl(\sum_{g,h}e_ge^\dagger_{h}\otimes I_d\otimes E_g^h\Bigr)(M\otimes I_{d'})\Bigl(\sum_{g',h'}e_{h'}e_{g'}^\dagger\otimes I_d\otimes E_{g'}^{h'}\Bigr).
\end{align*}
Thus $\rk(M')\leq \rk(M\otimes I_{d'})=d'\rk(M)$. Therefore, 
\begin{equation}\label{rankineq}
\cH_f^\C(G)\leq \frac{\rk(M')}{dd'}\leq \frac{d'\rk(M)}{dd'}\leq \frac{\rk(M)}{d}.
\end{equation}
Since \eqref{rankineq} holds for all $d$-representation $M$ of $H$ over $\C$, we conclude $\cH_f^\C(G)\leq \cH_f^\C(H)$. 

To prove that $\cH_{\smash{f}}^{\smash{\R}}$ is $\leq_q$-monotone, one follows the above proof with real instead of complex $d$-representations and one uses the fact that the projectors~$E_g^h$ can be chosen to be real matrices (Remark~\ref{remquant}).

Finally, we point out that $\bX(\cG,\leq_q)$ is a proper subset of $\bX(\cG,\leq)$. {Namely, $\overline{\chi}_f\in\bX(\cG,\leq)$ but $\overline{\chi}_f\not\in\bX(\cG,\leq_q)$. This is due to the fact that there exists a graph $G$ such that $\overline{\chi}_q(G)<\overline{\chi}_f(G)$~\cite{manvcinska2016quantum}, and $\max_{\phi\in\bX(\cG,\leq_q)}\phi(G)=\lim_{n\to\infty}\sqrt[n]{\overline{\chi}_q(G^{\boxtimes n})}\leq\overline{\chi}_q(G)$. On the other hand,}
it is known that $\Theta(G)$ can be strictly smaller than $\Theta_*(G)$ for some graph $G$~\cite{Leung2012,Briet19227}. 
More precisely, Bri{\"e}t, Buhrman and Gijswijt proved in \cite{Briet19227}  that for any odd prime $p$ such that there exists a Hadamard matrix of size $4p$, there exists a graph $G$ satisfying $\Theta(G)\leq \cH^{\F_p}(G)<\Theta_*(G)$. We note that the proof in~\cite{Briet19227} of $\cH^{\F_p}(G)<\Theta_*(G)$ in fact shows that $\cH^{\F_p}(G)<\Theta_q(G)$. \changed{The key observation is the following claim: If $\overline{G}$ has an orthonormal representation $U=(u_g\in\R^d:g\in V(G))$ in dimension~$d$, and $\overline{G}$ has $M$ disjoint $d$-cliques, then $\Theta_q(G) \geq M$. The proof is rather straightforward: Let the cliques be denoted by $C_1,\dots, C_M$. Take
\[
E_i^g = \begin{cases} u_gu_g^T & \textnormal{if $g \in C_i$}\\
0 & \textnormal{if $g \not\in C_i$.}
\end{cases}
\]
It is easy to see that $(E_i^g:g\in V(G),\, i\in V(\overline{K_M}))$ satisfies the conditions for the inequality $\overline{K_M}\leq_q G$. Thus $\Theta_q(G)\geq\alpha_q(G)\geq M$. This proves the claim.}

It is known that if $n$ is odd and there exists a Hadamard matrix of size $n+1$, then there exists a graph $G_n$ whose complement has an $n$-dimensional orthonormal representation and $|V(G_n)|/n^2$ disjoint cliques of size $n$~\cite{Briet19227}. 
Thus $\Theta_q(G_n)\geq |V(G_n)|/n^2$. On the other hand, it has been proved that the Haemers bound over some finite field $\F_p$ on $G_n$, $\cH^\F(G)$, can be strictly smaller than~$|V(G_n)|/n^2$~\cite{Briet19227}. Since $\cH_f^\F(G)\leq \cH^\F(G)$ for any field $\F$, we conclude that~$\cH_f^{\smash{\F_p}}\not\in\bX(\cG,\leq_q)$ for such odd prime $p$. 
\end{proof}

\begin{remark}
It is not hard to adjust the above proof to show that the fractional Haemers bound for any field extension of $\R$ belongs to $\bX(\cG,\leq_q)$. We show that $\cH_f^\R$ (resp.~$\xi_f^\R$) and $\cH_f^\C$ (resp.~$\xi_f$) are actually the same parameter. Moreover, one may naturally define a real projective rank $\xi_f^\R$ by requiring that the projectors in the definition of $\xi_f$ are real. Again, we show that $\xi_f^\R$ is equal to projective rank.
\end{remark}

\begin{proposition}\label{prop: real and complex}
$\cH_f^\R(G)=\cH_f^\C(G)$ and $\xi_f^\R(G)=\xi_f(G)$ for all graphs $G$.
\end{proposition}
\begin{proof}
The following lemma is readily verified.
\begin{lemma}\label{lemma: complex to real rank}
Let $E\in M(n,\C)$.
Define the real matrix
\[
R(E) = \begin{pmatrix}\Re(E)&\Im(E)\\-\Im(E)&\Re(E)\end{pmatrix} \in M(2n, \R).
\]
Then $\rk(R(E)) = 2\rk(E)$.
\end{lemma}

For the fractional Haemers bound, clearly $\cH_f^\C(G)\leq\cH_f^\R(G)$ since every real matrix is a complex matrix, and its rank over $\R$ equals the rank over $\C$. We prove $\cH_f^\R(G)\leq\cH_f^\C(G)$ by proving that, for every $d$-representation $M$ of $G$, there exists a $2d$-representation $M'$ of $G$, such that $\rk(M')\leq 2\rk(M)$. Assume $G$ has $n$ vertices. Write $M$ in the block matrix form
\begin{equation}
M=
\begin{bmatrix}
M_{1,1}&\cdots &M_{1,n}\\
\vdots& &\vdots\\
M_{n,1}&\cdots& M_{n,n}
\end{bmatrix}\in M(nd,\C),
\end{equation}
where $M_{i,i}=I_d$ for $i\in[n]$, and $M_{i,j}= M_{j,i}=0$ if $\{i,j\}\in E$. Let $M'$ be the $2nd\times 2nd$ real matrix of the form $M'=\begin{bmatrix}\Re(M)&\Im(M)\\-\Im(M)& \Re(M)\\\end{bmatrix}$. On the other hand, let $M_{i,j}'=\begin{bmatrix}\Re(M_{i,j})&\Im(M_{i,j})\\-\Im(M_{i,j})& \Re(M_{i,j})\\\end{bmatrix}$ and denote 
\begin{equation}
M''=
\begin{bmatrix}
M'_{1,1}&\cdots &M'_{1,n}\\
\vdots& &\vdots\\
M'_{n,1}&\cdots& M'_{n,n}
\end{bmatrix}\in M(2nd,\R).
\end{equation}
It is clear that $2\rk(M)\geq \rk(M')$ and $M''\in\cM_\R^{2d}(G)$. We show that $M''$ can be transformed to $M'$ by some row and column permutations of the blocks, which will not influence the rank. We first sort the columns, resulting that the first block row of the first $n$ block columns is $\Re(M_{1,1}),\ldots,\Re(M_{1,n})$ and the last $n$ block column is $\Im(M_{1,1}),\ldots,\Im(M_{1,n})$. Then we sort the rows, resulting that the first block column of the first $n$ block rows is $\Re(M_{1,1}),\ldots,\Re(M_{n,1})$ and the last $n$ block column is $\Im(M_{1,1}),\ldots,\Im(M_{n,1})$. Denote the matrix of these two permutations by $S$ and $T$, it is easy to check that $SM''T=M'$ (In fact, $T=S^T$). Thus $\rk(M'')=\rk(M')\leq 2\rk(M)$, and $\cH_f^\R(G)\leq\cH_f^\C(G)$ follows.

For the projective rank, let $(E_g\in M(d,\C):g\in V(G))$ be a $d/r$-representation of $G$. From Lemma~\ref{lemma: complex to real rank} follows that $(R(E_g) : g\in V(G))$
is a $2d/2r$-representation. We conclude that $\xi_f^\R(G)\leq\xi_f(G)$. On the other hand, every real~$d/r$-rep\-re\-sen\-ta\-tion is also a complex $d/r$-representation. Therefore, $\xi_f(G)\leq\xi_f^\R(G)$.
\end{proof}

{Whether $\cH_f^\R$ and $\overline{\xi}_f$ belong to $\bX(\cG,\leq_*)$ remains unknown. In particular, we do not see how to adapt the proof to show that they are monotones with respected to the entanglement-assisted cohomomorphism preorder. }

Recall that $\leq_q$ can be obtained from $\leq_*$ by restricting to the use of maximally entangled state and projective measurements in the zero-error information transmission setting~\cite{manvcinska2016quantum}. An open problem in quantum zero-error information theory is to show maximally entangled state is also necessary to achieve the maximal entanglement-assisted Shannon capacity~\cite{manvcinska2016quantum}. Namely, 
\begin{conjecture}\label{conj: theta q}
$\Theta_q(G)=\Theta_*(G)$ for all graphs $G$.
\end{conjecture} 

The original proof of Haemers~\cite{haemers1979some} shows that taking $G$ to be the complement of the Schl\"afli graph, $\cH^\R(G)\leq 7<9=\vartheta(G)$. By Theorem~\ref{thm: element in leq q}, we know that $\Theta_q$ and $\vartheta$ are not the same parameters, which immediately implies the following:
\begin{corollary}
Conjecture~\ref{conj: theta} and~\ref{conj: theta q} cannot both be true. In other words, there exists a graph $G$, such that either $\Theta_*(G)<\vartheta(G)$ or $\Theta_q(G)< \Theta_*(G)$.
\end{corollary}

\section{Dual characterization of entanglement-assisted Shannon capacity of non-commutative graphs}\label{sec: dual characterization nc graphs}
In this section, we focus on the fully quantum setting: the entanglement-assisted Shannon capacity of nc-graphs. We point out that Strassen's theory of asymptotic spectra seems not applicable to the (unassisted) Shannon capacity of nc-graphs. We discuss this further in Appendix~\ref{sec: unassisted}.

Recall that the map 
\[
\cG\to \cS : G \mapsto S_G\coloneqq\linspan\{\ket{g}\!\bra{g'}:g=g'\in V(G)~{\rm or}~\{g,g'\}\in E(G)\}
\]
is an injective semiring homomorphism. We prove that this homomorphism behaves well with respect to the entanglement-assisted cohomomorphism preorders on $\cG$ and $\cS$.
\begin{lemma}\label{prop: reduce to EA preorder graphs}
For any graphs $G, H\in \cG$, \changed{we have that}
$G\leq_* H$ if and only if $S_G\leq_*S_H$.
\end{lemma}
\begin{proof}
Let $|V(G)|=n$ and $|V(H)|=m$.

$(\Leftarrow)$ Assume there exist a positive definite $\sigma\in\cD(\C^d)$ and $E=\{E_i\}_i\subseteq\cL(\C^n\otimes\C^d,\C^m)$ satisfying $\sum_{i}E_i^\dagger E_i= I_{nd}$ and $E(S_G^\perp\otimes \sigma)E^\dagger \perp S_H$. 
Let $\sigma = \sum_{x=1}^d\lambda_x\Proj{\psi_x}$ be the spectral decomposition of $\sigma$ and let $\ket{\Omega}=\sum_{x=1}^d\sqrt{\lambda_x}\ket{\psi_x}\ket{x}\in\C^d\otimes\C^d$ be a purification of $\sigma$. Let
\begin{equation}
\rho_g^h=\sum_i(\bra{h}E_i\otimes I_{d})(\Proj{g}\otimes\Proj{\Omega})(E_i^\dagger\ket{h}\otimes I_{d})\in\cL(\C^d)
\end{equation}
for each $g\in V(G)$ and $h\in V(H)$.
Let $\rho\coloneqq\sum_{x=1}^d\lambda_x\Proj{x}$.
One verifies directly that every $\rho_g^h$ is positive semidefinite. We first prove that $\sum_{h\in V(H)}\rho_g^h=\rho$ for all $g\in V(G)$. Note that
\begin{equation}
 \begin{split}
 \sum_{h\in V(H)}\rho_g^h&=\sum_{h\in V(H)}\sum_{i} \bigl(\bra{h}E_i\otimes I_{d} \bigr)\bigl(\Proj{g}\otimes\Proj{\Omega}\bigr)\bigl(E_i^\dagger\ket{h}\otimes I_{d}\bigr)\\
 &=\sum_{h\in V(H)}\sum_{i,x,y}\sqrt{\lambda_x\lambda_y}\bigl(\bra{h}E_i (\Proj{g}\otimes\ket{\psi_x}\!\bra{\psi_y})E_i^\dagger\ket{h}\bigr)\otimes \ket{x}\!\bra{y}\\
&=\sum_{i,x,y}\sqrt{\lambda_x\lambda_y}\Tr\bigl(E_i(\Proj{g}\otimes\ket{\psi_x}\!\bra{\psi_y})E_i^\dagger\bigr)\otimes \ket{x}\!\bra{y}\\
&=\sum_{i,x,y}\sqrt{\lambda_x\lambda_y}\bigl(\bra{g}\otimes\bra{\psi_y}\bigr)E_i^\dagger E_i\bigl(\ket{g}\otimes\ket{\psi_x}\bigr)\otimes \ket{x}\!\bra{y}\\
&=\sum_x\lambda_x\Proj{x},
 \end{split}
 \end{equation} 
where the last equality holds since $\sum_i E_i^\dagger E_i= I$ and $\langle\psi_y|\psi_x\rangle=0$ for all $x\neq y\in[d]$. Now we are left to prove that $\rho_g^h\rho_{g'}^{h'}=0$ for all $g\neq g', \{g,g'\}\not\in E(G)$ and $\{h,h'\}\in E(H)$ or $h=h'$. Note that $\rho_g^h\rho_{g'}^{h'}$ equals
\begin{equation*}
\begin{split}
\rho_g^h\rho_{g'}^{h'} &= \sum_{i,j}(\bra{h}E_i\otimes I_{d})(\Proj{g}\otimes\Proj{\Omega})(E_i^\dagger\ket{h}\!\bra{h'}E_j\otimes I_{d})(\Proj{g'}\otimes\Proj{\Omega})(E_j^\dagger\ket{h'}\otimes I_{d})\\
=&\sum_{\substack{i,j\\x,y,z,w}}\sqrt{\lambda_x\lambda_y\lambda_z\lambda_w}\bra{h}E_i (\Proj{g}\otimes\ket{\psi_x}\!\bra{\psi_y})E_i^\dagger\ket{h}\!\bra{h'}E_j(\Proj{g'}\otimes\ket{\psi_z}\!\bra{\psi_w})E_j^\dagger\ket{h'}\otimes \langle y|z\rangle\ket{x}\!\bra{w}\\
=&\sum_{\substack{i,j\\x,y,w}}\sqrt{\lambda_x\lambda_w}\lambda_y\bra{h}E_i (\Proj{g}\otimes\ket{\psi_x}\!\bra{\psi_y})E_i^\dagger\ket{h}\!\bra{h'}E_j(\Proj{g'}\otimes\ket{\psi_y}\!\bra{\psi_w})E_j^\dagger\ket{h'}\ket{x}\!\bra{w}\\
=&\sum_{i,j,x,w}\sqrt{\lambda_x\lambda_w}(\bra{h}E_i(\ket{g}\otimes\ket{\psi_x}))(\Tr(E_j(\ket{g'}\!\bra{g}\otimes\sigma)E_i^\dagger\ket{h}\!\bra{h'}))((\bra{g'}\otimes\bra{\psi_w})E_j^\dagger\ket{h'})\ket{x}\!\bra{w},\\
\end{split}
\end{equation*}
where the last equality holds since 
\begin{equation}\label{eq: inner product 2}
\sum_{y}\lambda_y(\bra{g}\otimes\bra{\psi_y})E_i^\dagger\ket{h}\!\bra{h'}E_j(\ket{g'}\otimes\ket{\psi_y})=\Tr(E_j(\ket{g'}\!\bra{g}\otimes\sigma)E_i^\dagger\ket{h}\!\bra{h'}).
\end{equation}
Recall that $E(S_G^\perp\otimes \sigma) E\perp  S_H$, where  $S_G^\perp=\{\ket{g}\!\bra{g'}: g\neq g', \{g,g'\}\not\in E(G)\}$ and $S_H=\{\ket{h}\!\bra{h'}:\{h,h'\}\in E(H)~\textnormal{or}~h=h'\in V(H)\}$. Equation~\eqref{eq: inner product 2} equals $0$ when $\{g,g'\}\not\in E(G)$ and $h=h'\in V(H)$ or $\{h,h'\}\in E(H)$. We conclude that $\rho_g^h\rho_{g'}^{h'}=0$ for $g\neq g', \{g,g'\}\not\in E(G)$ and $h=h'\in V(H)$ or $\{h,h'\}\in E(H)$. 

$(\Rightarrow)$ Assume $G\leq_* H$. There exist $d\in\N$, a positive definite matrix $\rho\in M(d,\C)$ and positive semidefinite matrices $(\rho_g^h\in M(d,\C):g\in V(G),\, h\in V(H))$ such that $\sum_{h\in V(H)}\rho_g^h=\rho$ for all $g\in V(G)$ and $\rho_g^h\rho_{g'}^{h'}=0$ if $g\neq g', \{g,g'\}\not\in E(G)$ and $\{h,h'\}\in E(H)$ or $h=h'$. We shall prove that there exist a positive definite $\sigma\in\cD(\C^d)$ and $E=\linspan\{E_i\}_i\subseteq\cL(\C^n\otimes\C^d,\C^m)$ satisfying $\sum_i E_i^\dagger E_i=I_{nd}$, such that $E(S_G^\perp\otimes \sigma)E^\dagger \perp S_H$. We need the following lemma, which we prove for the convenience of the reader.
\begin{lemma}[\cite{HUGHSTON199314,Spekkens:2002:OCA:2011417.2011421}]\label{lemma: POVM}
Let $\rho_1,\dots,\rho_\ell\in M(d,\C)$ be a collection of positive semidefinite matrices which sum up to a positive definite matrix $\rho\in M(d,\C)$. Then there exist $\ket{\Omega}\in\C^d\otimes\C^d$ and a POVM $\{A_1,\dots,A_\ell\}\subseteq M(d,\C)$, i.e.\ a collection of positive semidefinite matrices that add up to the identity, such that $\rho_k=\Tr_1((A_k\otimes I)\Proj{\Omega})$.
Namely, let $\rho=\sum_{i=1}^{d}\lambda_i\Proj{\psi_i}$ be the spectral decomposition of $\rho$, so $\lambda_i>0$ for $i\in[d]$ and $\{\ket{\psi_1},\dots,\ket{\psi_{d}}\}$ forms an orthonormal basis of $\C^d$. Let $\ket{\Omega}=\sum_{i=1}^{d}\sqrt{\lambda_i}\ket{\psi_i}\otimes\ket{\psi_i}\in\C^d\otimes \C^d$ and let $A_k=\rho^{-1/2}\rho_k^T\rho^{-1/2}$.
\end{lemma}
\begin{proof}
We have
\[
A_k=\rho^{-1/2}\rho_k^T\rho^{-1/2}=\sum_{i,j=1}^{d}\frac{1}{\sqrt{\lambda_i\lambda_j}}\Proj{\psi_i}\rho_k^T\Proj{\psi_j}=\sum_{i,j=1}^{d}\frac{1}{\sqrt{\lambda_i\lambda_j}}\bra{\psi_i}\rho_k^T\ket{\psi_j}\ket{\psi_i}\!\bra{\psi_j}
\]
for $k\in[\ell]$. %
Moreover,
\begin{equation*}
\begin{split}
\Tr_1((A_k\otimes I)\Proj{\Omega})&=\sum_{i,j,x,y=1}^{d}\frac{\sqrt{\lambda_x\lambda_y}}{\sqrt{\lambda_i\lambda_j}}\bra{\psi_i}\rho_k^T\ket{\psi_j}\Tr_1((\ket{\psi_i}\!\bra{\psi_j}\otimes I)(\ket{\psi_x}\!\bra{\psi_y}\otimes\ket{\psi_x}\!\bra{\psi_y}))\\
&=\sum_{i,j}\bra{\psi_i}\rho_k^T\ket{\psi_j}\ket{\psi_j}\!\bra{\psi_i}=\rho_k.
\end{split}
\end{equation*}
This proves Lemma~\ref{lemma: POVM}.
\end{proof}
Following Lemma~\ref{lemma: POVM}, we define the pure state $\ket{\Omega}=\sum_{i=1}^{d}\sqrt{\lambda_i}\ket{\psi_i}\otimes\ket{\psi_i}\in \C^d\otimes\C^d$, where $\rho=\sum_{i=1}^{d}\lambda_i\Proj{\psi_i}$ is the spectral decomposition of $\rho$. Then, for every $g\in V(G)$, there exists a POVM $(A_g^h=\rho^{-1/2}(\rho_g^h)^T\rho^{-1/2}:h\in V(H))$, indexed by $h\in V(H)$, such that $\rho_g^h=\Tr_{1}((A_g^h\otimes I)\Proj{\Omega})$. 
For $g\neq g', \{g,g'\}\not\in E(G)$ and $\{h,h'\}\in E(H)$ or $h=h'$, $\rho_g^h\rho_{g'}^{h'}=0$ implies $(\rho_g^h)^T(\rho_{g'}^{h'})^T=0$, thus 
\begin{equation}\label{eq: Agh rho Ag'h'}
A_g^h\rho A_{g'}^{h'}=(\rho^{-1/2}(\rho_g^h)^T\rho^{-1/2})\rho(\rho^{-1/2}(\rho_{g'}^{h'})^T\rho^{-1/2})=\rho^{-1/2}(\rho_g^h)^T(\rho_{g'}^{h'})^T\rho^{-1/2}=0.
\end{equation}
Since $A_g^h$ is positive semidefinite for all $g\in V(G)$ and $h\in V(H)$, there is a spectral decomposition $A_g^h=\sum_{x}\mu_{x}\Proj{\phi^{g,h}_{x}}$, with $\mu_x>0$ for all possible $x$. By Equation~\eqref{eq: Agh rho Ag'h'}, we know that $\bra{\phi^{g,h}_{x}}\rho\ket{\phi^{g',h'}_{y}}=0$ for all possible $x,y$ and $g\neq g', \{g,g'\}\not\in E(G)$ and $\{h,h'\}\in E(H)$ or $h=h'$. Let $M_g^h=\sum_{x}\sqrt{\mu_x}\ket{\psi_x}\bra{\phi^{g,h}_{x}}$. We have $A_g^h=(M_g^h)^\dagger (M_g^h)$ and 
\begin{equation}\label{eq Mgh rho Mg'h'}
M_g^h \rho (M_{g'}^{h'})^{\dagger}=0,~\forall g\neq g', \{g,g'\}\not\in E(G)~{\rm and}~\{h,h'\}\in E(H)~{\rm or}~h=h'.
\end{equation}
Let $E=\linspan\{E_{i,g,h}=\ket{h}(\bra{g}\otimes\bra{\psi_i}M_g^h):g\in V(G),~h\in V(H),~i\in[d]\}\subseteq\cL(\C^n\otimes\C^d,\C^m).$ 
Note that 
\[
\sum_{i,g,h}E_{i,g,h}^\dagger E_{i,g,h}=\sum_{i,g,h} \Proj{g}\otimes(M_g^h)^\dagger \Proj{\psi_i} M_g^h=\sum_{g,h} \Proj{g}\otimes A_g^h=I_n\otimes I_d,
\]
where the second equality holds since $\{\ket{\psi_1},\dots,\ket{\psi_{d}}\}$ forms an orthonormal basis of $\C^d$ and since $A_g^h=(M_g^h)^\dagger (M_g^h)$, the last equality holds since $\sum_{h\in V(H)} A_g^h= I_d$ for all $g\in V(G)$. We also claim that $E(S_G^\perp\otimes\rho)E^\dagger\perp S_H$. We have
\[
E(S_G^\perp\otimes \rho)E^\dagger=\linspan\{\bra{\psi_i}(M_g^h\rho M_{g'}^{h'})\ket{\psi_j}\ket{h}\!\bra{h'}:\{g,g'\}\not\in E(G),~h,h'\in V(H),~i,j\in[d]\}.
\]
By Equation~\eqref{eq Mgh rho Mg'h'}, we know that $E(S_G^\perp\otimes \rho)E^\dagger$ is at most spanned by those operators $\ket{h}\!\bra{h'}$ with $h\neq h', \{h,h'\}\not\in E(H)$. This immediately implies $E(S_G^\perp \otimes \rho)E^\dagger\perp S_H$, since $S_H^\perp=\{\ket{h}\!\bra{h'}:\{h,h'\}\not\in E(H)\}$. This proves Lemma~\ref{prop: reduce to EA preorder graphs}.
\end{proof}

Now we prove that the entanglement-assisted cohomomorphism preorder $\leq_*$ (Definition~\ref{def: homomorphic preorder of quantum channels}) on nc-graphs is a Strassen preorder. %

\begin{lemma}\label{lemma: EA homo preorder quantum strassen}
For any nc-graphs $S\subseteq\cL(A)$, $S'\subseteq\cL(A')$, $T\subseteq\cL(B)$ and $T\subseteq\cL(B')$, \changed{it holds that}
\begin{enumerate}[label=\upshape(\roman*)]
\item $S\leq_* S$
\item if $S\leq_* T$ and $T\leq_* T'$, then $S\leq_* T'$
\item $\overline{\cK_m} \leq_* \overline{\cK_n}$ if and only if $m \leq n$
\item if $S \leq_* T$ and $S'\leq_* T'$, then $S\oplus S' \leq_* T\oplus T'$ and $S \otimes S' \leq_* T\otimes T'$
\item if $T\neq \overline{\cK_0}$, then there is an $r \in \N$ with $S \leq_* \overline{\cK_r} \otimes T$.
\end{enumerate}
\end{lemma}
\begin{proof}
(i) We know $S\leq T$ implies $S\leq_*T$ (Lemma~\ref{preorderrel}). %
Clearly $S\leq S$ holds by taking $E=\linspan\{I\}$. Therefore, also $S \leq_* S$.

(ii) Let a positive definite $\rho\in\cD(A_0)$ and $E=\linspan\{E_i\}_i\subseteq\cL(A\otimes A_0,B)$ be given by $S\leq_* T$, and a positive definite $\sigma\in\cD(B_0)$ and $F=\linspan\{F_j\}_j\subseteq\cL(B\otimes B_0,B')$ be given by $T\leq_* T'$. To see $S\leq_* T'$, take $\tau=\rho\otimes\sigma\in\cD(A_0\otimes B_0)$ and $F'=\linspan\{F_j(E_i\otimes I_{B_0})\}_{i,j}\subseteq\cL(A\otimes A_0\otimes B_0, B')$. We have
\begin{equation}
F'(S^\perp\otimes\tau){F'}^\dagger=F(E(S^\perp\otimes\rho)E^\dagger\otimes\sigma)F^\dagger \subseteq F(T^\perp\otimes\sigma)F^\dagger\perp T',
\end{equation}
where the inequality holds since $E(S^\perp\otimes\rho)E^\dagger\perp T$ by $S\leq_* T$, and the orthogonality relation is given by $T\leq_* T'$.

(iii) %
By Lemma~\ref{prop: reduce to EA preorder graphs}, $\overline{\cK_n}\leq_*\overline{\cK_m}$ is equivalent to $\overline{K_n}\leq_*\overline{K_m}$, which is equivalent to $m\leq n$ by Lemma~\ref{lemma: EA homo preorder strassen}.

{(iv) We only need to show that, if $S \leq_* T$, then for any $S'$ we have $S\oplus S' \leq_* T\oplus S'$ and $S \otimes S' \leq_* T\otimes S'$. Since if these hold, take $S \leq_* T$ and $S'\leq_* T'$, we have $S\oplus S'\leq_* T\oplus S'\leq_* T\oplus T'$ and $S\otimes S'\leq_* T\otimes S'\leq_* T\otimes T'$.\footnote{Recall that $S\oplus T$ (resp. $S\otimes T$) and $T\oplus S$ (resp. $T\otimes S$) are isomorphic up to unitary transformation, or more precisely, a permutation of the basis.}
Let $\rho\in\cD(A_0)$ be a positive definite matrix of size $d$ and $E=\linspan\{E_i\}_{i\in[m]}\subseteq\cL(A\otimes A_0,B)$ be given by $S\leq_* T$ such that $E(S^\perp \otimes \rho)E^\dagger\perp T$ and $\sum_{i=1}^m E_i^\dagger E_i=I_{A\otimes A_0}$. 

We first show $S\oplus S' \leq_* T\oplus S'$. Let 
\[
\hat{E}=\linspan\{\hat{E}_i:i\in[m+d],~\hat{E}_i=E_i\oplus 0~\text{if}~i\in[m],\hat{E}_{i}=0\oplus I_{A'}\otimes \bra{i-m}~\text{if}~i\in[m+d]\setminus[m]\}.
\]
where $\{\ket{1},\dots,\ket{d}\}$ denotes the computational basis of $A_0$. It is clear that $\hat{E}\subseteq\cL((A\otimes A_0\oplus A'\otimes A_0,B\oplus A')\cong \cL((A\oplus A')\otimes A_0,B\oplus A')$. Note that 
\[
\sum_{i=1}^{m+d}\hat{E}_i^\dagger\hat{E}_i=\sum_{i=1}^{m}\hat{E}_i^\dagger\hat{E}_i+\sum_{j=1}^d\hat{E}_{m+j}^\dagger\hat{E}_{m+j}=I_{A\otimes A_0}\oplus 0+0\oplus I_{A'}\otimes I_{A_0}=I_{(A\oplus A')\otimes A_0}.
\]
We only need to prove that $\hat{E}({(S\oplus S')}^\perp \otimes \rho)\hat{E}^\dagger\perp T\oplus S'$, which is equivalent to ${(S\oplus S')}^\perp \otimes \rho\perp \hat{E}^\dagger (T\oplus S')\hat{E}$. For $X\in T$ and $Y\in S'$, we know that
\[
\hat{E}_k^\dagger (X\oplus Y)\hat{E}_\ell=\left\{
\begin{array}{ll}
      E_k^\dagger X E_\ell\oplus 0 & k,\ell\in [m]\\
      0\oplus Y\otimes \ket{k-m}\!\bra{\ell-m} & k,\ell\in [m+d]\setminus[m]\\
      0& \textnormal{otherwise.} \\
\end{array}
\right.
\]
Note that for $Z=\begin{bmatrix}Z_{1,1}&Z_{1,2}\\Z_{2,1}&Z_{2,2}\end{bmatrix}\in (S\oplus S')^\perp$, we have that $\Tr(Z_{1,1}C)=0$ for any $C\in S$ and $\Tr(Z_{2,2}D)=0$ for any $D\in S'$. For every $Z=\begin{bmatrix}Z_{1,1}&Z_{1,2}\\Z_{2,1}&Z_{2,2}\end{bmatrix}\in (S\oplus S')^\perp$, $X\in T$ and $Y\in S'$, we have
\begin{equation}\label{eq: direct sum trace}
\Tr((Z\otimes \rho)(\hat{E}_k^\dagger (X\oplus Y)\hat{E}_\ell))=\left\{
\begin{array}{ll}
      \Tr((Z_{1,1}\otimes \rho)E_k^\dagger X E_\ell) & k,\ell\in [m]\\
      \bra{\ell-m}\rho\ket{k-m}\Tr(Z_{2,2} Y) & k,\ell\in [m+d]\setminus[m]\\
      0& \textnormal{otherwise.} \\
\end{array}
\right.
\end{equation}
It is clear that, by the choice of $E=\linspan\{E_i\}_{i\in[m]}$, $\Tr((Z_{1,1}\otimes \rho)E_k^\dagger X E_\ell)=0$ since $Z_{1,1}\in S^\perp$ and $E_k^\dagger X E_\ell\perp S^\perp\otimes\rho$ for any $X\in T$, $k,\ell\in[m]$. It is also clear that $\bra{\ell-m}\rho\ket{k-m}\Tr(Z_{2,2} Y)=0$ for any $k,\ell\in[m+d]\setminus[m]$ since $Y\in S'$ and $Z_{2,2}\perp S'$. These conclude that Equation~\eqref{eq: direct sum trace} is $0$ for any $Z\in (S\oplus S')^\perp$, $X\in T$, $Y\in S'$ and $k,\ell\in[m+d]$, and $\hat{E}({(S\oplus S')}^\perp \otimes \rho)\hat{E}^\dagger\perp T\oplus S'$ follows by linearity.

Now we show $S \otimes S' \leq_* T\otimes S'$. Let 
\[
\tilde{E}=\linspan\{\tilde{E}_i=E_i\otimes I_{A'}\}_{i\in [m]}\subseteq \cL(A\otimes A_0\otimes A',B\otimes A').
\]
It is clear that $\sum_{i=1}^m\tilde{E}_i^\dagger \tilde{E}_i=I_{A\otimes A_0}\otimes I_{A'}=I_{A\otimes A_0\otimes A'}$. We only need to prove that $\tilde{E}({(S\otimes S')}^\perp \otimes \rho)\tilde{E}^\dagger\perp T\otimes S'$, which is equivalent to ${(S\otimes S')}^\perp \otimes \rho\perp \tilde{E}^\dagger(T\otimes S')\tilde{E}$. 

\changed{We first establish the following claim: For any nc-graphs $S\subseteq \cL(A)$ and $S'\subseteq\cL(A')$ we have that $(S\otimes S')^\perp$ equals $\linspan\{S^\perp\otimes S',S\otimes (S')^\perp, S^\perp\otimes (S')^\perp\}$. To see this, let $T=\linspan\{S^\perp\otimes S',S\otimes (S')^\perp, S^\perp\otimes (S')^\perp\}$. It is easy to see that $T\subseteq (S\otimes S')^\perp$. To see that $T$ and $(S\otimes S')^\perp$ are the same subspace of $\cL(A\otimes A')$, we prove that they have the same dimension.
Let $\dim(S)=a_1$, $\dim(\cL(A))=n_1$, $\dim(S')=a_2$ and $\dim(\cL(A'))=n_2$. Then $\dim((S\otimes S')^\perp)=n_1n_2-a_1a_2$. Observe that $S^\perp\otimes S'$, $S\otimes (S')^\perp$ and $S^\perp\otimes (S')^\perp$ are pairwise orthogonal. We have $\dim(T)=\dim(S^\perp\otimes S')+\dim(S\otimes (S')^\perp)+\dim(S^\perp\otimes (S')^\perp)=(n_1-a_1)a_2+a_1(n_2-a_2)+(n_1-a_1)(n_2-a_2)=n_1n_2-a_1a_2=\dim((S\otimes S')^\perp)$. This concludes the proof of the claim.}

Note that $\tilde{E}^\dagger(T\otimes S')\tilde{E}=E^\dagger T E\otimes S'\subseteq (S^\perp \otimes \rho)^\perp\otimes S'$, where the last inclusion holds by the choice of $E$. We shall prove that $(S^\perp \otimes \rho)^\perp\otimes S'\perp {(S\otimes S')}^\perp \otimes \rho$, where 
\[
{(S\otimes S')}^\perp \otimes \rho=\linspan\{S^\perp\otimes S'\otimes \rho,S\otimes (S')^\perp\otimes \rho,S^\perp\otimes (S')^\perp\otimes \rho\}
\]
\changed{holds due to the above claim}. For any $X\in S^\perp\otimes S'\otimes \rho\cong (S^\perp\otimes\rho)\otimes S'$, $X\perp (S^\perp \otimes \rho)^\perp\otimes S'$. For any $Y\in S\otimes (S')^\perp\otimes \rho$ and $Z\in S^\perp\otimes (S')^\perp\otimes \rho$, it is clear that $Y,Z\perp (S^\perp \otimes \rho)^\perp\otimes S'$. $(S^\perp \otimes \rho)^\perp\otimes S'\perp {(S\otimes S')}^\perp \otimes \rho$ follows by linearity, and we conclude that ${(S\otimes S')}^\perp \otimes \rho\perp \tilde{E}^\dagger(T\otimes S')\tilde{E}$.
}

(v) We show that for any $S,T\neq\overline{\cK_0}$, there exists an $r\in \N$ such that $S\leq_*\overline{\cK_{r}}\otimes T$. We first claim that $S\leq \cI_n:=\linspan\{I_n\}\subseteq\cL(\C^n)$ for $n=\dim(A)$. {This can be done by simply taking $E=\linspan\{I_n\}$, since by definition of an nc-graph we have $I_n \in S$.} We then show that, for any $n\in\N$, $\cI_n\leq_*\overline{\cK_{n^2}}$. Let $E_{i,j}=\ket{\Phi_{i,j}}\!\bra{i,j}$ for all $i,j\in\{0,\dots,n-1\}$, where $\{\ket{i}\otimes\ket{j}:i,j\in\{0,\dots,n-1\}\}$ is the computational basis of $\C^n\otimes\C^n$ and
\begin{equation}
\ket{\Phi_{i,j}}:=\frac{1}{\sqrt{n}}\sum_{k=0}^{n-1}(X(i)Z(j)\ket{k})\otimes\ket{k},
\end{equation}
where $X(i)\ket{k}=\ket{i+k\mod{n}}$ and $Z(j)\ket{k}=\exp(i2\pi jk/n)\ket{k}$, is the $(i,j)$-th element of the Bell basis of $\C^n\otimes\C^n$ (cf.~\cite[page 114]{wilde_2017}). Take $\rho= I_n$ and $E=\{E_{i,j}:i,j\in\{0,\dots,n-1\}\}\subseteq \cL(\C^{n^2})$. Note that $\cI_n^\perp=\{X\in\cL(\C^n):\Tr(X)=0\}$, and $X\otimes I_n\perp \Proj{\Phi_{i,j}}$ for all $i,j\in\{0,\dots,n-1\}$, since
\[
\Tr((X^\dagger\otimes I_n)\Proj{\Phi_{i,j}})=\Tr(X^\dagger\Tr_2(\Proj{\Phi_{i,j}}))=\Tr(X^\dagger)=0.
\]
This implies that $(\cI_n^\perp\otimes I_n)\perp E^\dagger \overline{\cK_{n^2}} E$, which is equivalent to $E(\cI_n^\perp\otimes I_n)E^\dagger\perp \overline{\cK_{n^2}}$. Thus we conclude that $S\leq_*\overline{\cK_{N^2}}$ by transitivity. We derive that $S\leq_*\overline{\cK_{N^2}}\otimes\overline{\cK_1}\leq_*\overline{\cK_{N^2}}\otimes T$ if $T\neq\overline{\cK_0}$, which concludes the proof.
\end{proof}

Let $\bX(\cS,\leq_*)$
be the asymptotic spectrum of nc-graphs with respect to the entanglement-assisted cohomomorphism preorder, i.e.
\begin{equation}\label{eq: X(S,leq_*)}
\bX(\cS,\leq_*)=\{\phi\in\Hom(\cS,\R_{\geq 0}):\forall S,T\in \cS,~S\leq_* T~\Rightarrow~\phi(S)\leq\phi(T)\}.
\end{equation}

Together with Theorem~\ref{thm: asymptotic subrank}, we obtain the following dual characterization of the entanglement-assisted Shannon capacity of nc-graphs.
\begin{theorem}\label{thm: ea capacity of quantum channels}
Let $S$ be an nc-graph. Then
\[
\Theta_*(S)=\min_{\phi\in\bX(\cS,\leq_*)}\phi(S).
\]
\end{theorem}
Recall that there exists an injective semiring homomorphism $\iota: \cG \to \cS$ mapping the graph~$G$ to the nc-graph~$S_G$ (Lemma~\ref{injhom}) such that $G \leq_* H$ if and only if $S_G \leq_* S_H$ (Lemma~\ref{prop: reduce to EA preorder graphs}). 
By Theorem~\ref{restr} this implies that there exists a surjection from $\bX(\cS,\leq_*)$ to $\bX(\cG, \leq_*)$ via $\iota$.
\begin{theorem}\label{surje}
The map
\[
\bX(\cS,\leq_*) \rightarrow \bX(\cG, \leq_*) : \phi \mapsto \phi \circ \iota
\]
is surjective.
\end{theorem}
Since $\vartheta \in \bX(\cG, \leq_*)$, we know by Theorem~\ref{surje} that there exists a function in $\bX(\cS,\leq_*)$ that restricts to $\vartheta$.
Indeed, Duan, Severini and Winter in \cite{duan2013} introduced the quantum Lov\'asz theta function~$\tilde{\vartheta}$, which has these properties. This is currently the only element in $\bX(\cS,\leq_*)$ that we know~of.
\begin{theorem}\label{quantum lovasz theta}
We have
\[
\tilde{\vartheta} \in \bX(\cS,\leq_*).
\]
Moreover, $\tilde{\vartheta}(S_G)=\vartheta(G)$ for any graph $G$.
\end{theorem}
{\begin{proof}
	For any two nc-graphs $S\subseteq \cL(A)$ and $S'\subseteq\cL(A')$, $\tilde{\vartheta}(S\otimes S')=\tilde{\vartheta}(S)\tilde{\vartheta}(S')$~\cite[Corollary~10]{duan2013} and $\tilde{\vartheta}(S\oplus S')=\tilde{\vartheta}(S)+\tilde{\vartheta}(S')$~\cite[Proposition 17 in the arXiv version]{duan2013}. In~\cite{duan2013}, they also proved $\tilde{\vartheta}(S_G)=\vartheta(G)$ for any graph $G$, thus $\tilde{\vartheta}(\overline{\cK_n})=\vartheta(\overline{K_n})=n$ for any $n\in\N$. Lastly, $\tilde{\vartheta}(S)\leq\tilde{\vartheta}(S')$ if $S\leq_* S'$ has been shown in~\cite[Theorem 19]{stahlke2016}. This concludes the proof.
	\end{proof}
}

{Although we do not know any other explicit element in $\bX(\cS,\leq_*)$, we know that there must be at least one more element in $\bX(\cS,\leq_*)$ besides $\tilde{\vartheta}$. This is due to the separation result in~\cite{8260566}: We know that there exist (a family of) noncommutative graphs $S$ satisfying that $\Theta_*(S)<\vartheta(S)$.%
}

We summarize our knowledge of the asymptotic spectra of graphs and noncommutative graphs in~Figure~\ref{fig: relations between asymptotic spectrum}. %

\begin{figure}[ht]
\center
\begin{tikzpicture}[scale=1]
\draw[fill=gray!20] (9,5) circle (4 and 2.8);
\node at (9,8.3) {\small $\bX(\cG,\leq)$}; 

\draw[fill=green!20] (7.3,5) circle (2.2 and 1.6);
\node at (7.2,6.25) {\small $\bX(\cG,\leq_q)$}; 

\draw[fill=orange!20] (6.5,5) circle (1.4 and 1);
\node at (6.5,5.6) {\small $\bX(\cG,\leq_*)$}; 

\draw[fill=yellow!20] (0.5,5) circle (3 and 2);
\node at (0.5,7.5) {\small $\bX(\cS,\leq_*)$};

\node at (6,5) {$\bullet$};
\node at (6,4.7) {\footnotesize $\vartheta$};
\node at (2.8,5) {$\bullet$};
\node at (2.8,4.7) {\footnotesize $\tilde{\vartheta}$};
\node[red!90] at (8.5,5.9) { $\bullet$};
\node[black!75] at (9.1,6.25) { \small$\bullet$};
\node[black!60] at (9.7,6.5) {\small $\bullet$};
\node[black!45] at (10.3,6.7) {\footnotesize $\bullet$};
\node[black!40] at (10.9,6.8) {\scriptsize$\bullet$};
\node[black!30] at (11.5,6.85) {\tiny$\bullet$};

\node at (8.5,5.5) {\footnotesize $\cH_f^\R$};
\node at (9.9,6.1) {\footnotesize $\cH_f^{\F}$};
\node at (10.3,7.1) {\footnotesize $\cH_f^{\F_p}$};

\node at (12,5) {$\bullet$};
\node at (12,4.7) {\footnotesize $\overline{\chi}_f$};

\node[red] at (7.9,3.8) {$\bullet$};
\node at (8.25,3.8) {\footnotesize $\overline{\xi}_f$};

\draw[dashed,blue] (1,6.98) -- (6.6,6) node [midway, above, sloped] {\small $\kappa$};
\draw[dashed,blue] (1,3.02) -- (6.6,4);
\draw[dashed,blue] (2.8,5) -- (6,5); 

\end{tikzpicture}

\caption{Relations among asymptotic spectra of graphs and non-commutative graphs with different preorder. The fractional Haemers bound provide an infinite family of elements in $\bX(\cG,\leq)$. We don't know whether the red elements belong to smaller asymptotic spectra or not. It is also open whether $\bX(\cG,\leq_*)=\bX(\cG,\leq_q)$.}
\label{fig: relations between asymptotic spectrum}
\end{figure}
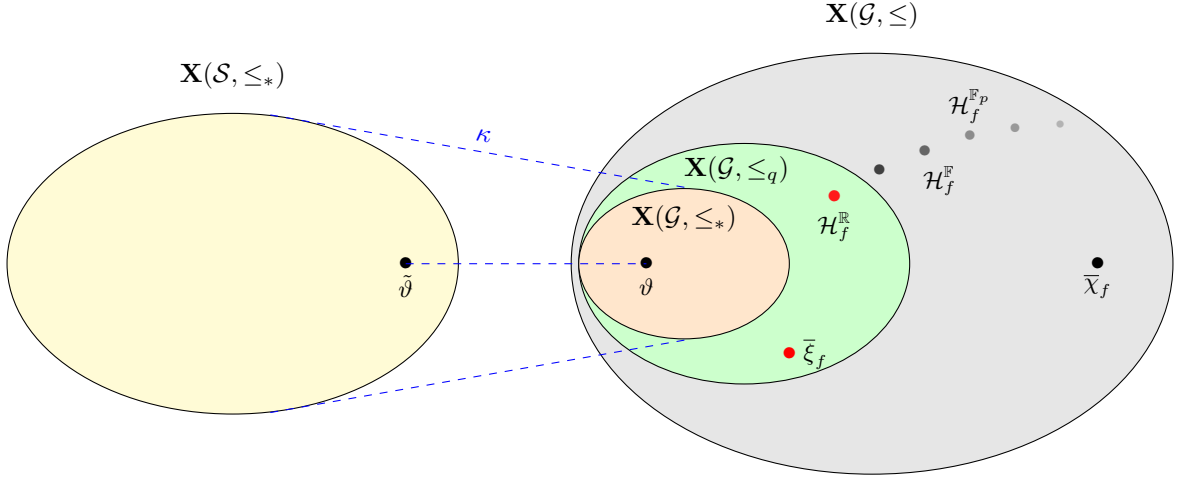

\paragraph{Acknowledgment}
The authors thank Jop Bri\"{e}t, Monique Laurent, Ronald de Wolf, Michael Walter, Christian Schaffner, Xin Wang, and Runyao Duan for useful discussions. YL thanks Runyao Duan for hosting his visit to the Institute of Quantum Computing, Baidu Inc. YL is supported by ERC Consolidator Grant 615307-QPROGRESS. JZ was a PhD student at Centrum Wiskunde \& Informatica while writing part of this paper. JZ was supported by NWO (617.023.116) and the QuSoft Research Center for Quantum Software.
This material is based upon work supported by the National Science Foundation under Grant No.~DMS-1638352.

\appendix

\section{Proof of Lemma~\ref{prop: nc-graph independent numbers}}\label{sec: proof of 2.8}

\begin{proof}[Proof of Lemma~\ref{prop: nc-graph independent numbers}]
Let $S\leq \cL(A)$ be an nc-graph.

(i) We show that $\overline{\cK_n}\leq S$ if and only if there is a size-$n$ independent set of $S$. 

Suppose $\Proj{\psi_1},\dots,\Proj{\psi_n}$ is a size-$n$ independent set of $S$. 
Let $E_i=\ket{\psi_i}\!\bra{i}$ for $i\in[k]$. Then $E_i \ket{\ell}\!\bra{\ell'} E_j^\dagger=\delta_{i,\ell}\delta_{j,\ell'}\ket{\psi_\ell}\!\bra{\psi_{\ell'}}$ for all $\ell\neq \ell'\in[n]$. 
We compute the inner product of $E_i \ket{\ell}\!\bra{\ell'} E_j^\dagger$ and $X$. We have $\Tr(E_j\ket{\ell'}\!\bra{\ell}E^\dagger_i X)=\delta_{i,\ell}\delta_{j,\ell'}\Tr(\ket{\psi_{\ell'}}\!\bra{\psi_\ell} X)$, for all $i,j$, $\ell\neq \ell'\in[n]$ and $X\in S$. For all $i\neq \ell$ or $j\neq \ell'$, the previous equation equals $0$ since $\delta_{i,\ell}=0$ or $\delta_{j,\ell'}=0$, and otherwise $\Tr(\ket{\psi_{\ell'}}\!\bra{\psi_\ell} X) = 0$ for all $X\in S$ as $\{\Proj{\psi_1},\dots,\Proj{\psi_k}\}$ forms an independent set of $S$. This concludes that $E_i \ket{\ell}\!\bra{\ell'} E_j^\dagger\perp X$ for all $i,j$, $\ell\neq \ell'\in[n]$ and $X\in S$, which implies $\overline{\cK_n}\leq S$.

On the other hand, suppose $\overline{\cK_n}\leq S$. Then there exist $\cE: \cL(\C^n)\to \cL(A)$ with Choi--Kraus operators $\{E_1,\dots,E_\ell\}\subseteq \cL(\C^n,A)$, such that $E \overline{\cK_n}^\perp E^\dagger\perp S$. The condition $E \overline{\cK_n}^\perp E^\dagger\perp S$ is equivalent to $E^\dagger S E \subseteq \overline{\cK_n}= \linspan\{\ket{i}\!\bra{i}:i\in[n]\}$. Since $I_A \in S$, we then have $E_j^\dagger E_j\in\linspan\{\ket{i}\!\bra{i}:i\in[n]\}$ for any $j\in[\ell]$. Thus we know $E_j^\dagger E_j$ is diagonal. 
By the singular value decomposition, 
 we can write $E_j=\sum_{x_j}\sqrt{\lambda_{x_j}^j}\ket{\psi_{x_j}^j}\!\bra{v_{x_j}^j}$, where $\lambda_{x_j}^j>0$ since $E_j^\dagger E_j$ is positive semidefinite, $\ket{v^j_{x_j}}\in\{\ket{1},\dots,\ket{k}\}$ for all possible $x_j$ and $\langle \psi_{x_j}^j|\psi_{y_j}^j\rangle=0$ for all possible $x_j\neq y_j$. Then for $X\in S$,
\[
E_j^\dagger X E_\ell=\sum_{x_j,y_\ell}\sqrt{\lambda^j_{x_j}\lambda^\ell_{y_\ell}}\bra{\psi_{x_j}^j}X\ket{\psi_{y_\ell}^\ell}\ket{v_{x_j}^j}\!\bra{v_{y_\ell}^\ell}\in \linspan\{\ket{i}\!\bra{i}:i\in[n]\},
\]
which implies $\bra{\psi_{x_j}^j}X\ket{\psi_{y_j}^\ell}=0$ if $\ket{v_{x_j}^j}\neq \ket{v_{y_\ell}^\ell}$ for all $X\in S$. Note that $\sum_{j}E_j^\dagger E_j=I_{\C^n}$. Thus $\linspan\{\ket{v^j_{x_j}}\}_{j,x_j}=\C^n$. This then guarantees that we can find a size-$n$ independent set of $S$.

(ii) We show that $\overline{\cK_n}\leq_* S$ if and only if there is a size-$n$ entanglement-assisted independent set of $S$. 

Suppose $\ket{\Omega}\in A_0\otimes B_0$ and $\cE_1,\dots,\cE_k$ form an entanglement-assisted independent set of $S$.  Let $\rho=\Tr_{B_0}(\Proj{\Omega})$ and $\cE:\cL(\C^n\otimes A_0)\to\cL(A)$ be the quantum channel which maps $\ket{i}\!\bra{i}\otimes \sigma$ to $\cE_i(\sigma)$ for all $i\in[k]$ and $\sigma\in\cD(A_0)$. The Choi--Kraus operators of $\cE$ can be written as $\{\bra{i}\otimes E_{i,j}\}_{i\in[n],j}$, where $\{E_{i,j}\}_{j}$ are the Choi--Kraus operators of $\cE_i$. We obtain that $E(\overline{\cK_n}^\perp\otimes \rho) E^\dagger=\linspan\{E_{i,j}\rho E_{k,\ell}^\dagger:i\neq k\in[n],~j,\ell\} \perp S$. 
We conclude $\overline{\cK_n}\leq_* S$.

Suppose $\overline{\cK_n}\leq_* S$. Then there exist a positive definite $\rho\in\cD(A_0)$ and a quantum channel $\cE:\cL(\C^n\otimes A_0)\to\cL(A)$ with Choi--Kraus operators $\{E_i\}_i$ such that $E(\overline{\cK_n}^\perp\otimes\rho)E^\dagger \perp  S$. Let $\cE_i(\rho)=\cE(\Proj{i}\otimes \rho)$ for $i\in[n]$ and let $\ket{\Omega}\in A_0\otimes B_0$ be a purification of $\rho$. The Choi--Kraus operator of $\cE_i$ can be written as $\{E_{i,j}=E_j(\ket{i}\otimes I_{A_0})\}_j\subseteq\cL(A_0, A)$. Then for $i\neq i'\in[n]$ and $j,j'$, $E_{i,j}^\dagger\rho E_{i',j'}=E_j(\ket{i}\!\bra{i'}\otimes\rho)E_{j'}^\dagger\in E(\overline{\cK_n}^\perp\otimes\rho)E^\dagger$, thus $\linspan\{E_{i,j}^\dagger\Tr_{B_0}(\Proj{\Omega}) E_{i',j'}:i\neq i'\in[n],~j,j'\}\perp S$. %
We conclude $\{\ket{\Omega},\cE_1,\dots,\cE_n\}$ is an entanglement-assisted independent set of $S$.
\end{proof}

\section{The unassisted Shannon capacity of nc-graphs}\label{sec: unassisted}
In this section, we discuss the unassisted Shannon capacity of nc-graphs. 
The unassisted Shannon capacity may not admit a dual characterization by its asymptotic spectrum. We first note that the cohomomorphism preorder on nc-graphs becomes the cohomomorphism preorder on graphs when restricting from nc-graphs to graphs.
\begin{lemma}[{\cite[Theorem~8]{stahlke2016}}]\label{prop: reduce to preorder graphs}
For any graphs $G, H\in \cG$ \changed{we have that}
$G\leq H$ if and only if $S_G\leq S_H$.
\end{lemma}
The cohomomorphism preorder on nc-graphs has the following properties.
\begin{lemma}\label{lemma: homo preorder quantum strassen}
For any nc-graphs $S\subseteq\cL(A)$, $S'\subseteq\cL(A')$, $T\subseteq\cL(B)$ and $T'\subseteq\cL(B')$ and $n, m \in \N$, we have
\begin{enumerate}[label=\upshape(\roman*)]
\item $S\leq S$
\item if $S\leq T$ and $T\leq T'$, then $S\leq T'$
\item $\overline{\cK_m} \leq \overline{\cK_n}$ if and only if $m \leq n$
\item if $S \leq T$ and $S'\leq T'$, then $S\oplus S' \leq T\oplus T'$ and $S \otimes S' \leq T\otimes T'$
\end{enumerate}
\end{lemma}
\begin{proof}
(i) We see that $S\leq S$ by taking $E=\linspan\{I_A\}$ in Definition~\ref{def: homomorphic preorder of quantum channels}.

(ii) Let $E=\linspan\{E_i\}_i\subseteq\cL(A,B)$ be given by $S\leq T$, and $F=\linspan\{F_j\}_j\subseteq\cL(B,B')$ be given by $T\leq T'$. To see $S\leq T'$, take $F'=\linspan\{F_jE_i\}_{i,j}\subseteq\cL(A, B')$. We have
\begin{equation}
F'S^\perp{F'}^\dagger=F(E S^\perp E^\dagger) F^\dagger \subseteq FT^\perp F^\dagger\perp T',
\end{equation}
where the inequality holds since $ES^\perp E^\dagger\perp T$ by $S\leq T$, and the last orthogonality relation is given by $T\leq T'$.

(iii) %
By Lemma~\ref{prop: reduce to preorder graphs}, $\overline{\cK_n}\leq \overline{\cK_m}$ is equivalent to $\overline{K_n}\leq \overline{K_m}$, which is equivalent to $m\leq n$.

(iv) Let $E=\linspan\{E_i\}_i\subseteq\cL(A,B)$ be given by $S\leq T$, and $F=\linspan\{F_j\}_j\subseteq\cL(A',B')$ be given by $S'\leq T'$. Let $E'=\linspan\{E_i\oplus 0\}_i \cup \{0\oplus F_j\}_j \subseteq \cL(A\oplus A',B\oplus B')$, where $(E_i\oplus F_j)(\ket{\psi}_A\oplus\ket{\psi'}_{A'})=E_i\ket{\psi}_A\oplus F_j\ket{\psi'}_{A'})$ for all $i,j$ and $\ket{\psi}_{A}\in A$ and $\ket{\psi'}_{A'}\in A'$. One readily verifies that $E'(S\oplus S')^\perp{E'}^\dagger\perp T\oplus T'$. To see $S \otimes S' \leq_* T\otimes T'$, Let $E'=\linspan\{E_i\otimes I_{A'},I_{A}\otimes F_j\}_{i,j}\subseteq\cL(A\otimes A',B\otimes B')$. One readily verifies that $E'(S\otimes S')^\perp{E'}^\dagger\perp T\otimes T'$.
\end{proof}

Recall the following property of the entanglement-assisted cohomomorphism preorder $\leq_*$. If $T\neq \overline{\cK_0}$, then there is an $r \in \N$ with $S \leq_* \overline{\cK_r} \otimes T$. The next example shows that the cohomomorphism preorder $\leq$ does not have this property, thus cannot be a Strassen preorder.
\begin{example}
Let $S=\cI_2$ and $T=\overline{\cK_1}=\C$. For any $r\in\N$ \changed{it holds that} $S\not\leq \overline{\cK_r}\otimes T$.	
\end{example}
\begin{proof}
Assume $\cI_2\leq \overline{\cK_{r}}$. Let $E=\linspan\{E_i\}_i\leq \cL(\C^2,\C^r)$ satisfy $E \cI_2^\perp E^\dagger \perp \overline{\cK_r}$. 
Note that $ES^\perp E^\dagger\perp \overline{\cK_{r}}$ implies $E^\dagger\overline{\cK_{r}}E\subseteq S$, since $\Tr(E_iX^\dagger E_j^\dagger Y)= \Tr(E_i^\dagger Y^\dagger E_jX)$ implies $E_j^\dagger YE_i\perp X$ for all $E_i,E_j\in E$, $X\in S^\perp$ and $Y\in \overline{\cK_r}$. We obtain that $E_i^\dagger\Proj{j}E_i\in \cI_2$ for all $i$ and $j\in[r]$. This is impossible since $E\neq 0$ and since the nonzero elements in $\cI_2$ have rank 2.
\end{proof}
The reason why $\cI_2\not\leq \overline{\cK_r}$ for every $r\in\N$ can be understood as: no classical channels can transmit even a single qubit. In the entanglement-assisted setting, this can be overcome by invoking the teleportation protocol, as mentioned in the proof of Lemma~\ref{lemma: EA homo preorder quantum strassen} (v).

\raggedright
\bibliographystyle{alphaurl}
\bibliography{all}

\end{document}